\documentclass{llncs}
\usepackage{latexsym}
\usepackage[utf8]{inputenc}

\usepackage{times}
\usepackage{graphicx}

\usepackage{amsmath,amsfonts,mathrsfs,amssymb,mathtools}
\usepackage{bbold}
\usepackage{stmaryrd}
\usepackage{color}

\usepackage{url}

\usepackage[caption=false]{subfig}

\usepackage{varwidth}

\renewenvironment{proof}{\paragraph{Proof} }{\hfill\qed}
\renewcommand{\paragraph}[1]{\noindent\textit{#1}.}
\newcommand{\BigO}[1]{\ensuremath{\operatorname{\mathcal{O}}\bigl(#1\bigr)}}

\begin{document}

%\begin{frontmatter}

% \author{Kevin Atighehchi}
% \ead{kevin.atighehchi@univ-amu.fr}
% \address{Aix-Marseille Université, Laboratoire d'Informatique 
% Fondamentale de Marseille, case 901, F13288 Marseille cedex 9, France}
% \author{Robert Rolland}
% \ead{robert.rolland@acrypta.fr}
% \address{Aix-Marseille Université, Institut de Mathématiques de 
% Marseille, case 907, F13288 Marseille cedex 9, France}

% \author{Kevin Atighehchi}
% \address{Aix-Marseille Université, Laboratoire d'Informatique 
% Fondamentale de Marseille, case 901, F13288 Marseille cedex 9, France}
% \email{kevin.atighehchi@univ-amu.fr}
% 
% 
% \author{Robert Rolland}
% \address{Aix-Marseille Université, Institut de Mathématiques de 
% Marseille, case 907, F13288 Marseille cedex 9, France}
% \email{robert.rolland@acrypta.fr}
% \date{\today}
% \keywords{}
%\subjclass[2000]{11T71, 94B05}

\title{Optimization of Tree Modes for Parallel Hash Functions: A Case Study}

\author{Kevin Atighehchi\inst{1} \and Robert Rolland\inst{2}}

\institute{Aix Marseille Univ, CNRS, LIF, Marseille, France\\
\email{kevin.atighehchi@univ-amu.fr}\\
\and
Aix Marseille Univ, CNRS, I2M, Marseille, France\\
\email{robert.rolland@acrypta.fr}}

\maketitle

%\maketitle

\begin{abstract} 
This paper focuses on parallel hash functions based on tree modes of operation for an inner Variable-Input-Length function. 
This inner function can be 
either a single-block-length (SBL) and \emph{prefix-free} MD hash function, or a sponge-based hash function. 
We discuss the various forms of optimality that 
can be obtained
when designing parallel hash functions based on trees where all leaves have the same depth.
The first result is a scheme which optimizes 
the tree topology in order to decrease the running time. 
Then, without affecting the optimal running time 
we show that we can slightly change the corresponding tree topology so as 
to minimize the number of required processors as well.
Consequently, the resulting scheme decreases in the first place the running 
time and in the second place the number of required processors. 
%The present work is of 
%independent interest if we consider the problem of parallelizing the evaluation of an expression
%where the operator used is neither associative nor commutative.
%Finally discuss the way to choose a tradeoff between these forms of optimality.
\end{abstract}

%\begin{keyword}

%\end{frontmatter}

\section{Introduction}\label{sec:intro}
%Asymmetric algorithms are in fact more costly than their symmetric equivalents and therefore they 
%require process a small input data size.
%The hash function is then used to generate a digest from the message.
%Once signed, this digest will serve as a reference value to check the integrity of the message.
%To be unable to forge a signature without knowing the private key, the signature algorithm
%must be existentially unforgeable and the hash function collision resistant.
% Une fonction de hachage cryptographique a de multiples applications dont la principale est d’être
% utilisée dans un algorithme de signature, pour compresser un message préalablement avant de le signer.
% Les algorithmes asymétriques sont en effet plus coûteux que leurs homologues sysmétriques et prennent donc
% en entrée une donnée de taille réduite. 
% La fonction de hachage génère donc un condensé à partir des
% données électroniques qui servira de valeur de référence pour vérifier leur
% intégrité.
% Pour qu’il soit difficile de forger une signature sans connaître la
% clé privée, l’algorithme de signature doit garantir une propriété d’infalsifiabilité et la fonction de hachage
% sous-jacente doit être résistante aux collisions et aux secondes préimages. 

A mode of operation for hashing is an algorithm iterating (operating)
an underlying function over parts of a message, under a particular composition method, 
%having a fixed input size (\emph{e.g.}, a compression function or a block cipher) 
in order to compute a digest. 
A hash function is obtained by applying such a mode to a concrete underlying function;
we then
say that the former is constructed on top of the latter.
%We then say that this hash function is constructed 
%on top of this underlying function.
Usually, when the purpose is to process messages of arbitrary length, the underlying function may be a fixed-input-length (FIL) compression function, 
a block cipher or a permutation. 
However, there may also be an interest in using a sequential (or serial) hash function as underlying function, in order to add to it other features, 
like coarse-grained parallelism.
%When the purpose is to add a better parallelism to a sequential hashing algorithm 
%that already processes messages of arbitrary length, the underlying function is a hash function.
%process messages of arbitrary lengths. 
The resulting hash function must satisfy the usual properties of pre-image resistance (given a digest value, it is hard to find 
any pre-image producing this digest value),
second pre-image resistance (given a message $m_1$, it is hard to find a second message $m_2$ which produces the same digest value),
and collision resistance (it is hard to find two distinct messages which produce the same digest value).
A sequential hash function can only use 
Instruction-Level Parallelism (ILP) and SIMD instructions \cite{GK12a,GK12b}.
% Une fonction de hachage est un algorithme itérant une fonction dont l’entrée est de taille fixe, une
% fonction de compression ou bien un algorithme de chiffrement par bloc, de façon à pouvoir traiter des
% messages de taille quelconque. Un mode de hachage séquentiel ne peut profiter que du 
% parallélisme au niveau instruction et des instructions SIMD \cite{GK12a,GK12b}.
A cryptographic hash function has numerous applications, the 
%principal among them 
main one is
its use
in a signature algorithm to compress a message before signing it.

The most well known sequential hashing mode is the Merkle-Damgård \cite{Dam90,Mer79} construction
which can only take advantage of the fine-grained parallelism of the operated compression function.
If such a low-level "primitive" can benefit from the Instruction-Level Parallelism, by using also
SIMD instructions, the outer algorithm iterating this building block could benefit from a coarse-grained 
parallelism. This parallelism can be employed in multithreaded implementations.
Suppose that we have a collision-free compression function taking as input
a fixed-size data, $f : \{0,1\}^{2N} \rightarrow \{0,1\}^N$. By using a balanced binary tree structure, 
Merkle and Damgård \cite{Dam90,Mer80} show that we can extend the domain 
of this function so that the new outer function, denoted
$H : \{0,1\}^* \rightarrow \{0,1\}^N$, has an arbitrary sized domain and is still collision-resistant.
Note that if the function $f$ is a sequential hash function, the purpose of this tree structure is merely the addition 
of coarse-grained parallelism.
% Le mode séquentiel de hachage le plus connu est certainement celui de Merkle-Damgård \cite{Dam90,Mer79},
% qui ne peut profiter que du parallélisme à grains fins de la fonction de compression utilisée. Si une telle
% primitive de bas niveau peut profiter d’un parallélisme au niveau instruction, en employant par ailleurs
% un jeu d’instructions SIMD, l’algorithme itérant cette brique de base pourrait lui profiter d’un parallélisme à gros grain,
% qui aurait alors un intérêt à être employé dans une implémentation multithreadée.
% Supposons qu’on dispose d’une fonction de hachage (ou de compression) résistante aux collisions et prenant en entrée
% une donnée de taille fixe, $f : \{0,1\}^{2N} \rightarrow \{0,1\}^N$. En utilisant une structure arborescente binaire
% équilibrée, comme la construction historique de Merkle et Damgård \cite{Dam90,Mer80}, nous pouvons étendre le domaine 
% de cette fonction de sorte que la nouvelle fonction, notée
% $H : \{0,1\}^* \rightarrow \{0,1\}^N$ , ait un domaine de taille aléatoire et soit toujours résistante aux collisions.

A construction using a balanced binary tree allows simultaneous processing of multiple parts of data 
at a same level of the tree, 
reducing the running time to hash the message from $\BigO{n}$ to $\BigO{\log n}$ 
if we have $\BigO{n}$ processors \cite{Dam90,Mer80}. If we want to further reduce
the amount of resources involved, we can use one of the following rescheduling techniques:
\begin{itemize}
 \item Each processor is assigned the processing of a subtree (in the data structure sense) having $\log n$ leaves. 
 There are approximatively $n/\log n$ such subtrees. 
The processing of the remaining ancestor nodes, at each remaining level of the tree, is distributed as fairly as possible
between the processors. An example is depicted in Figure \ref{resch_technique}.

 \item An alternative solution is, at each level of the tree, to distribute as fairly as possible the node computations among
$\BigO{n/\log n}$ processors.
\end{itemize}

\begin{figure}[htbp]
\begin{center}
%\scalebox{0.43}{
% \input{illustration_rescheduling.pspdftex} 
\includegraphics{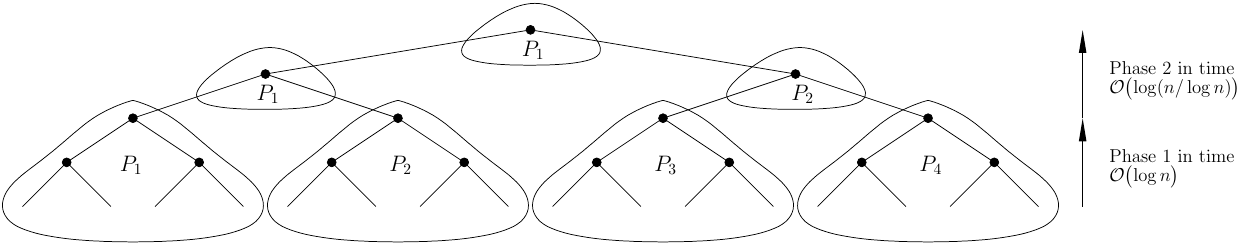}
%}
\caption{Example of the computation of the root node in $\BigO{\log n}$ time using $\BigO{n/\log n}$ processors. The message
to hash is of size $n=16$. In Phase 1, the computation of each hash subtree containing $4=\log_2 16$ leaves is assigned to each processor. The first 
subtree is assigned to processor $P_1$, the second one to processor $P_2$, the third one to processor $P_3$ and the last one to processor $P_4$. A fine-grained allocation
is then performed in Phase~2.}
\label{resch_technique}
\end{center}
\end{figure}

% Une construction en arbre binaire équilibré permet à de multiple parties des données sur un
% même niveau de l’arbre d’être traitées simultanément, réduisant ainsi le temps d’exécution pour hacher
% le message de $\BigO{n}$ à $\BigO{\log n}$ si l’on se donne $\BigO{n}$ processeurs. Si l'on désire réduire davantage 
% le nombre de ressources impliquées, on peut utiliser l'une des techniques de réordonnancement suivantes :
% \begin{itemize}
%  \item Chaque processeur se voit affecter le traitement d'un sous-arbre (au sens des structures de données) 
%  de hauteur $\log \log n$. Il y a approximativement $n/\log n$ tels sous-arbres. 
%  Le traitement des autres noeuds ancêtres, à chaque niveau restant de l'arbre, est distribué aussi équitablement que possible
%  entre les processeurs.
%  \item L'autre solution consiste, à chaque niveau de l'arbre, à distribuer aussi équitablement que possible les calculs des noeuds entre 
% $\BigO{n/\log n}$ processeurs.
% \end{itemize}
%qui consiste, à chaque niveau de l'arbre, à distribuer aussi équitablement que possible les calculs des noeuds entre 
%$\BigO{n/\log n}$ processeurs. 
The number of processors is then reduced by a factor $\log n$ and
the asymptotic running time is conserved (with, nevertheless, a multiplicative factor 2).
%avec $\BigO{n/\log n}$ processeurs.
In this paper we are not interested in tradeoffs between
the amount of used resources and the running time but instead we study optimal algorithms in finite distance.
More precisely, we determine the hash tree structures which give the best concrete (parallel) time complexity for finite message lengths.
% Le nombre de processeurs est alors diminué d'un facteur $\log n$ et
% le temps d'exécution asymptotique est conservé (avec néanmoins un facteur multiplicatif 2).
% %avec $\BigO{n/\log n}$ processeurs.
% Dans ce papier nous ne nous intéresserons pas aux compromis
% entre ressources utilisées et temps d'exécution mais nous étudierons plutôt des algorithmes optimaux en distance finie.

A tree structure is notably used in parallel hashing modes of Skein \cite{FLSWBKCW09}, BLAKE2~\cite{ANWW13} or MD6 \cite{RABCDEKKLRSSSTY08}. 
To give some examples, Skein uses a tree whose topology is controlled by the user thanks to three parameters:
the arity of base level nodes which is a power of two; the arity of other inner nodes,
which is also a power of two, and a last parameter limiting the height of the tree.
MD6 uses a full (but not necessarily perfect) quaternary tree, in the sense that an inner node has always four children.
Some fictive leaves or nodes padded with 0 are added so that a rightmost node has the correct number of children. 
Like Skein, MD6 offers a parameter which serves to limit the height of the tree.
% Une topologie en arbre couvrant la totalité des blocs du message est notamment utilisée pour le
% mode de hachage parallel de Skein \cite{FLSWBKCW09}, de Blake2 \cite{ANWW13} ou encore de MD6 \cite{RABCDEKKLRSSSTY08}. 
% Pour détailler quelques exemples, Skein emploie un arbre dont la topologie est contrôlée par l'utilisateur à l'aide de trois paramètres :
% l'arité des noeuds du niveau de base, qui est une puissance de deux; l'arité des autres noeuds internes,
% qui est aussi une puissance de deux, et un dernier paramètre limitant la hauteur de l'arbre.
% MD6 emploie un arbre quaternaire entier
% non nécéssairement parfait, dans le sens où un noeud interne possède toujours quatre fils.
% Des feuilles ou des noeuds supplémentaires composés de $0$ sont en effet rajoutés de sorte qu'un noeud est le bon nombre de fils.
% Tout comme Skein, MD6 propose à l'utilisateur un paramètre limitant la hauteur de l'arbre.
%ou un noeud feuille
%qui n’a pas de frère est promut au niveau suivant sans que celui-ci soit rehaché, jusqu’à ce qu’un frère soit trouvé.

Some proposals \cite{SS01,SS02,PS03} consider that a tree covering all the message blocks
is not a good thing, because the number of processors should not grow with the size of the message.
For instance, the domain extension parallel algorithm from Sarkar et al. \cite{SS01,SS02,PS03} uses a perfect binary tree
of processors, of fixed size. 
This perfect binary tree of compression/hash function calls can be seen as a big compression function, sequentially iterated
over large parts of the message. 
In other words only the nodes computations performed in the tree can be done in parallel.
The number of usable processors is a system parameter chosen by the issuer of the cryptographic form when hashing the message.
The value of this parameter has to be reused by the recipients, for instance when verifying a signature. 
Thus, this one limits the scalability and the potential speedup. 
In this paper we consider that the scalability and the potential speedup 
should be independent of the characteristics (configuration) of the transmitting computer.
% Certaines propositions \cite{SS01,SS02,PS03} considèrent qu’un arbre couvrant la totalité des blocs du message
% n’est pas une bonne chose, car le nombre de processeurs ne devrait pas grossir avec la taille du message.
% L’algorithme parallel d’extension de domaine de Sarkar et al. \cite{SS01,SS02,PS03} utilise un arbre binaire parfait
% de processeurs, de taille fixe. Cet arbre binaire parfait peut être vu comme un grosse fonction de compression, itérée séquentiellement 
% sur des parties de taille importante du message. 
% Autrement dit seuls les calculs effectués dans l'arbre sont faits en parallèle.
% Le nombre de processeurs utilisable est un paramètre du système défini par l'emétteur lors du hachage.
% Ce paramètre devant être réutilisé par les machines réceptrices, par exemple lors de la vérification d'une signature émise par l'émetteur, 
% celui-ci limite donc la mise à l’échelle et le parallélisme potentiel. Dans ce papier nous considérons que la mise
% à l'échelle et le parallélisme potentiel devraient être indépendants des caractéristiques de la machine émettrice.

Bertoni \textit{et al.} \cite{BDPV09,BDPV14_Sak} give sufficient conditions for a tree-based hash function to ensure its indifferentiability
from a random oracle. They propose several tree hashing modes for different usages.
For example we can make use of a tree of height 2, defined in the following way: we divide the message in as many parts (of roughly equal size) 
as there are processors so that each processor hashes each part, 
and then the concatenation of all the results is sequentially hashed by one processor. 
To divide the message in parts of roughly equal size, the algorithm needs to know in advance the size of the message.
Bertoni \textit{et al.} propose also a variant
which still makes use of a tree having two levels and a fixed number of processors, but this one interleaves 
the blocks of the message. This interleaving offers a number of advantages, as it allows an efficient parallel hashing of a streamed message, 
a fairly equal distribution of the data processed by each processor in the first level of tree (without prior knowledge of the message size), 
and a correct alignment of the data in the processors' registers. This kind of solution is suitable for multithreaded and SIMD implementations \cite{Gue14}.
In this paper we study theoretically optimal speedups, and, 
as a consequence, the message to hash is supposed to be already available.
% Bertoni et al. \cite{BDPV09,BDPV14_Sak} donnent des conditions suffisantes pour qu'une fonction de hachage employant un arbre 
% soit indifférenciable 
% d'un oracle aléatoire. Ils proposent également plusieurs modes de hachage en arbre dédiés à des usages différents.
% On peut par exemple employer un arbre de hauteur 2, défini de la façon suivante : on coupe un message en autant 
% de parties (de tailles à peu près égales) qu'il y a de processeurs de sorte que chaque processeur hache chacune des parties, 
% et la concaténation des résultats est haché séquentiellement par un des processeurs. Pour partager les données en parts à peu près égales, 
% l'algorithme doit connaître la taille du message à l'avance.
% En empruntant une idée de Gueron \cite{Gue14},
% ils proposent une variante qui emploie toujours un arbre à deux niveaux et un nombre fixe de processeurs, mais qui entrelace 
% les blocs du messages. Cet entrelacement a plusieurs intérêts. Il permet de hacher en parallèle et de manière efficace 
% un flux de données streamé, de partager à péu près équitablement la quantité de données traités par chacun des processeurs au premier
% niveau de l'arbre sans connaître à l'avance la taille du message, et enfin d'aligner correctement les données dans les registres du processeurs.

% Then, if we consider a tree of calls of
% this compression function, the computation of a node has a cost equals to the number of its children. 

Our concern in this paper is with hash tree modes using an underlying variable-input-length (VIL)
function that needs $l$ invocations of a lower level primitive to process a message of $l$ blocks, 
where a block and the hash output have the same size.
%To concretize such a complexity,
To make such a complexity concrete,
we choose to use a single-block-length (SBL) hash function as underlying function 
%\textit{i.e.} a hash (or compression) 
and to focus on the \emph{prefix-free} Merkle-Damg{\aa}rd construction from Coron \textit{et al.} \cite{CDMP05}. We make this choice for two reasons: 
first, we need a hash function whose mode of operation is proven indifferentiable from a random oracle (when its underlying primitive is assumed to be ideal).
Second, assuming that we have applied the \emph{prefix-free} encoding \cite{CDMP05} 
and another encoding \cite{BDPV09,BDPV14_Sak} to identify the type of input in the tree, it is possible to precompute a constant number of hash states,
making the aforementioned complexity possible.
Note that, even if we take this construction as an example, the possible use of a sponge-based function 
will be discussed.
% function that needs $l$ invocations of the
% underlying primitive to process a message of $l$ blocks, where a block and the hash output have the same size.
In this work, we aim to show that we can improve
the performance of a hash tree mode of operation by reworking the tree-structured circuit topology.
While we focus on the case of trees having all their leaves at the same depth,
%In particular, 
we are interested in minimizing the depth (parallel time)
of the circuit and the width (number of processors involved). This kind of work has been done for parallel 
exponentiation in finite fields \cite{Sti90,Gat91,AMV88a,LKPC05,WLLC06} where the multiplication operator is both associative and commutative.
In the case of parallel hashing, the considered operator can be a FIL (Fixed-Input Length) compression function. This is quite different since 
we do not have these two properties and we need to cope with other problems (the space consumption of a padding rule, a length encoding, 
or other information bits).
To the best of our knowledge, it is the first time that the problem of optimizing hash trees is addressed. The main interest of this paper is 
the methodology provided.
%The results are the followings:
The results can be presented as follows:
\begin{itemize}
 \item The first result is an algorithm which optimizes 
the tree topology in order to decrease the depth. We first show that a node arity greater than $5$ is not possible and then we prove
that we can construct such an optimal tree using exclusively levels of arity $2$ and $3$.
 \item Without affecting this optimal depth, 
we show that we can change the corresponding tree topology in order to decrease the width as much as possible. 
%This width is optimal for trees having all their leaves at the same level. 
In particular, we show that for some message lengths $l$, the width can be decreased to $\lceil l/5 \rceil$.
%  \item Observations are made about trees having their leaves at different levels, indicating that if our previous algorithm does not produce optimal solutions 
%  for this kind of trees, it probably produces near-optimal solutions.
 \item We also provide an algorithm which optimizes the number of processors 
 at each step of the hash computation. We prove that eleven tree topologies are possible.
 \item With the assumption that the message size is Pareto-distributed, we estimate the relative frequency of each tree topology 
 using the Monte Carlo method.
 \item Finally, we show that by using a SBL hash function as underlying function and by assuming a constant number of precomputed values, 
 these optimisations can be applied safely.
 %by using only $4$ different compression functions having the same running time.
\end{itemize}

Suppose that the processing of one block of the message by the underlying function costs one unit of time. 
A binary tree is not necessarily the structure which gives the best running time. Figure \ref{Arbre_exemple} shows two different tree topologies for
hashing a 6-block message. 
The binary tree depicted in (\ref{fig:sub1}) gives a (parallel) running time of $6$ units while the rightmost one with a different arity 
at each level, depicted in (\ref{fig:sub2}), gives a running time of 5 units. Furthermore, one may note that for messages of length less that 5 blocks, 
the use of the topology (\ref{fig:sub1}) has no utility compared to a purely sequential mode (\textit{i.e.} a completely degenerated binary tree).
%, which also provide a running time of~$6$.

\begin{figure}[h]
\centering
%\begin{subfigure}{.5\textwidth}
\subfloat[Non optimal tree]
  {
  \centering
  \includegraphics[scale=0.43]{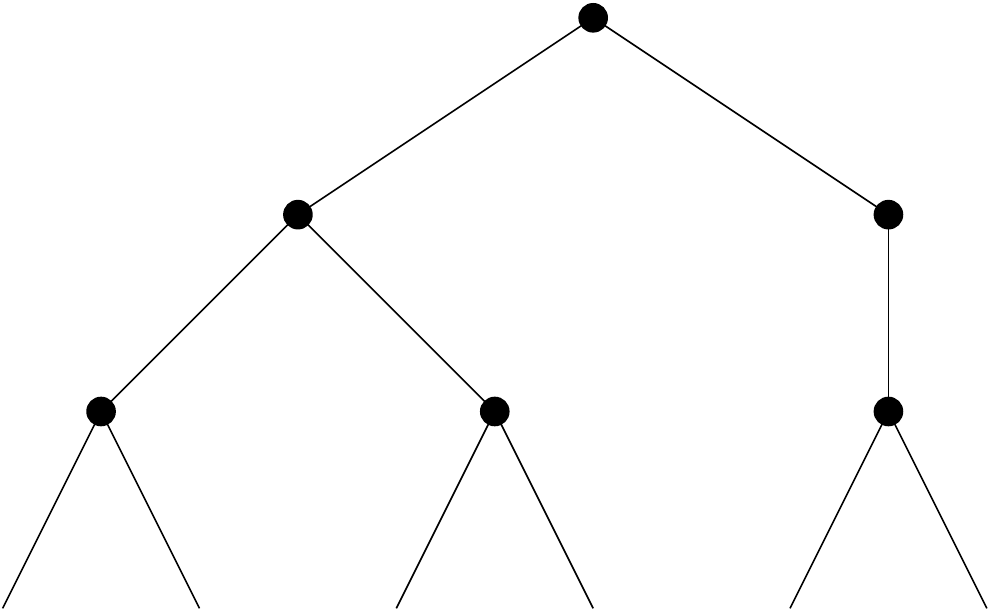}
  %\caption{Non optimal tree}
  \label{fig:sub1}
  }
\qquad\qquad
%\end{subfigure}%
%\begin{subfigure}{.5\textwidth}
\subfloat[Optimal tree]
  {
  \centering
  \includegraphics[scale=0.43]{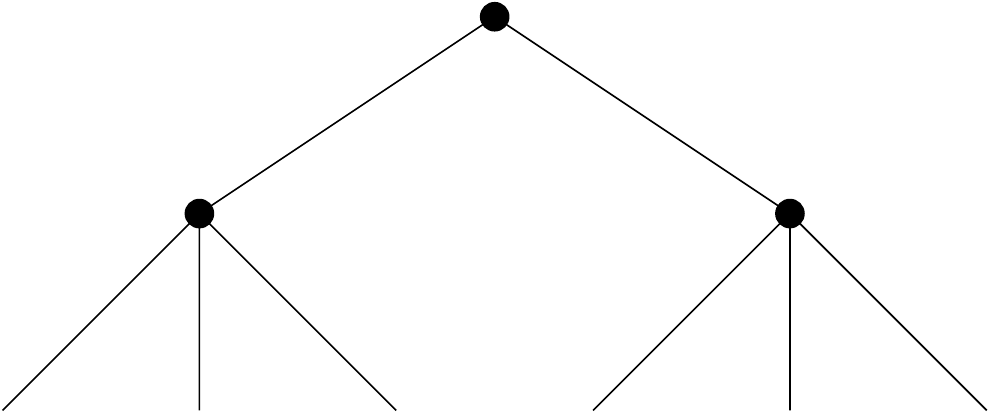}
  %\caption{Optimal tree}
  \label{fig:sub2}
  }
%\end{subfigure}
\caption{Tree hashing with a 6-block message. The hash tree on the left requires 2 units of time to process each level, while the one on the right requires 
3 units of time to process the base level and 2 units of time to process the root node.}
\label{Arbre_exemple}
\end{figure}

%%%%%%%%%%%%%%%%% partie incluant la sécurité
% In what follows, we suppose the use of a multitude of different compression functions, namely $f_{x,w} : \{0,1\}^{xN} \rightarrow \{0,1\}^N$ for $x \geq 2$ and $w \in \{B,I,R\}$. 
% For each node, the choice of the function to use to compress its children depends on 
% their number $x$ (the node arity) and on the position $w$ of this node. If this is a base level node then $w=B$. If this is an inner node which is not the root node, then
% $w=I$, otherwise $w=R$.
% We also assume that a compression function which compresses $x$ blocks of size $N$ bits has a computational cost of $x$ units.
% In other words, if we consider a tree of calls of
% this compression function, the computation of a node having $k$ children (\emph{i.e.} $k$ blocks) has a cost of $k$ units.
% For instance, the UBI compression function, used in the hash function family Skein \cite{FLSWBKCW09}, performs $x$ calls to the tweakable block cipher Threefish
% to compress a data of length $x$ blocks. Assuming a hash tree of height $h$ and $x_i$ the arity of level $i$ (for $i=1 \ldots h$), we define the parallel 
% running time to obtain the root node value as being $\sum_{i=1}^h x_i$.
%%%%%%%%%%%%%%%%%

% In what follows, we suppose the use of a multitude of different compression functions, namely $f_{x} : \{0,1\}^{xN} \rightarrow \{0,1\}^N$ for $x \geq 2$. 
% For each node, the choice of the function to use to compress its children depends on 
% their number $x$ (the node arity). 
In what follows, we suppose the use of variable-input-length (VIL) compression functions (or hash functions)
having a domain space $\{0,1\}^{xN}$  for $x \geq 2$ and a fixed length range space $\{0,1\}^N$.
% For each node, the choice of the function to use to compress its children depends on 
% their number $x$ (the node arity).
We also assume that such a function has an ideal computational cost of $x$ units when compressing $x$ blocks of size $N$ bits.
In other words, if we consider a tree of calls of
this function, the computation of a node having $k$ children (\emph{i.e.} $k$ blocks) has a cost of $k$ units.
Such a computational cost is realist. For instance, the UBI transformation function\footnote{UBI stands for Unique Block Iteration, the 
sequential operating mode used in Skein. The UBI transformation function refers to the application of this mode to the underlying tweakable
compression function, itself based on the tweakable block cipher Threefish.
%The Skein hash function
%is defined in a modular fashion. It is based on the UBI function, itself based on the Threefish tweakable block cipher.
}
used in the hash function family Skein \cite{FLSWBKCW09} performs $x$ calls to the tweakable block cipher Threefish
to compress a data of length $x$ blocks. Assuming a hash tree of height $h$ and $x_i$ the arity of level $i$ (for $i=1 \ldots h$), we define the parallel 
running time to obtain the root node value as being $\sum_{i=1}^h x_i$.

The paper is organized in the following way. In Section 2 we give background information and definitions.
% des informations générales sur la
% théorie des graphes et des fonctions de hachage. 
In Section 3, we first describe the approach to minimize the running time of a hash function. Then, we give an algorithm to
construct a hash tree topology which achieves the same optimal running time while requiring an optimal number of processors.
We also show that we can optimize the number of processors at each step of the hash computation. This leads to eleven possible
tree topologies, whose probability distribution is empirically analyzed in Section \ref{prob_dist}.
We propose in Section 5
a concrete tree-based hash function that safely implements these optimizations.
Finally, in Section 5,  we conclude the paper. %and discuss future works.
%to achieve this optimal running time.
%d'une fonction de hachage et enfin la section 4 donne un algorithme pour décroître le nombre de processeurs de façon presque optimale.

% \section{Background information and definitions}\label{subsec:backg}
\section{Preliminaries}\label{subsec:backg}

%%%%%%%%%%%%%%%%%%
\subsection{Tree structures}
%%%%%%%%%%%%%%%%%%

Throughout this paper, 
%(except when referring to security aspects) 
we use the 
convention\footnote{This corresponds to the convention used to describe Merkle trees. 
%The other (less frequent)
%convention is to define a node as being 
%a $f$-input.
} that a node is the result of
a function called on a data composed of the node's children.
A node value (or chaining value) then corresponds to an image by such a function and a child of this node can be either
an other image or a message block.
%and also possibly an auxiliary data (metadata).
We call a base level node a node at level $1$ pointing to the leaves representing message data blocks. The leaves (or leaf nodes) 
are then at level~$0$. Then, a tree of nodes of height $h$ has $h+1$ levels.
We define the arity of a level in the tree as being the greatest node arity in this level.

A $k$-ary tree is a tree where the nodes are of arity at most $k$. For instance, a tree with only one node of arity $k$ is said to be a $k$-ary tree. 
A full $k$-ary tree is a tree where all nodes have exactly $k$ children.
A perfect $k$-ary tree is a full $k$-ary tree where all leaves have the same depth.

\begin{figure}[htb]
\centering
%\begin{subfigure}[t]{.4\textwidth}
\subfloat[Ternary tree, or $(2,3)$-aries tree]
  {
  \centering
  \includegraphics[scale=0.50]{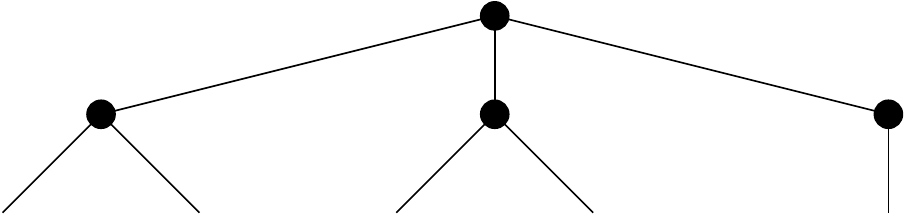}
  %\caption{Running time of an optimal tree (shown in blue) compared to a binary tree (in black)}
  \label{perf1}
  }
%\end{subfigure}%
\vspace{0.75cm}

%\begin{subfigure}[t]{.4\textwidth}
\subfloat[Full $(2,3)$-aries tree]
  {
  \centering
  \includegraphics[scale=0.50]{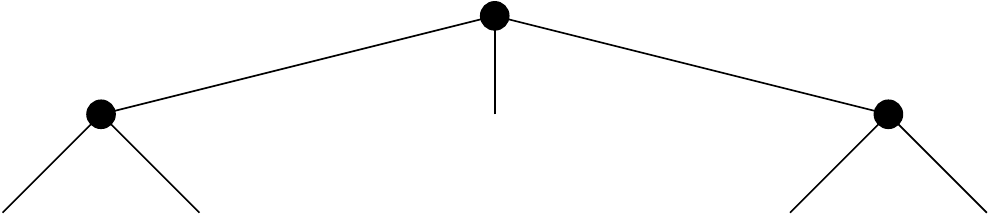}
  %\caption{Running time of an optimal tree (shown in blue) compared to a binary tree (in black)}
  \label{perf1}
  }
%\end{subfigure}%
%~~~~~
\vspace{0.75cm}

%\begin{subfigure}[t]{.4\textwidth}
\subfloat[Full and perfect ternary tree, also seen as a full and perfect $(3,3)$-aries tree]
  {
  \centering
  \includegraphics[scale=0.50]{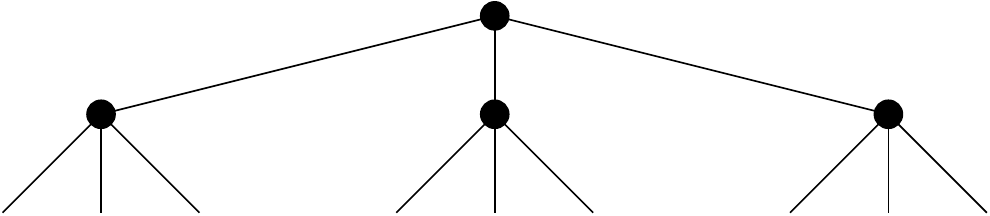}
  %\caption{Speed gain of an optimal tree compared to a binary tree}
  \label{perf2}
  }
%\end{subfigure}
\caption{Examples of trees}
\label{Tree_examples}
\end{figure}

\begin{sloppypar}
We also define other "refined" types of tree. 
We say that a tree is of arities $\{k_1,k_2, \ldots, k_n\}$ (we can call it a $\{k_1,k_2, \ldots, k_n\}$-aries tree) if 
it has $n$ levels (not counting level $0$) whose nodes at the first level are
of arity at most $k_1$, nodes at level $2$ are of arity at most $k_2$, and so on. We say that such a tree is full 
if all nodes at the first level have exactly $k_1$ children, all nodes at level $2$ have exactly $k_2$ children, and so on. 
As before, we say that such a tree is perfect if it is full and if all the leaves are at the same depth.
Some examples are depicted in Figure~\ref{Tree_examples}.
\end{sloppypar}

A tree of nodes has a corresponding tree of $f$-inputs having one less level. For instance, the simple tree of nodes of height 1 and arity 3 has 2 levels and 4 nodes: 
one root node and 3 children which are the message blocks. Its corresponding tree of $f$-inputs has a single level containing a single $f$-input. This $f$-input consists
of the concatenation of the message blocks and some meta-information bits. This representation is further defined in the following subsection \ref{three_conditions}.

\subsection{Security of tree-based hash functions}\label{three_conditions}

In this section (and in Subsection~\ref{sec_section}), we represent a hash tree as a tree of $f$-input,
assuming that a single inner function $f$ is operated in the outer hash function, denoted~$H$.
This is the representation adopted in~\cite{BDPV09,BDPV14_Suf} to prove the desired security properties.
%A tree of nodes then corresponds to a tree of $f$-inputs. 
We recall that a tree of $f$-inputs has one less level compared to a tree of nodes. 

A $f$-input is a finite 
sequence of bits from the following elements: message bits, chaining value bits (\emph{i.e.} bits coming from a $f$-image), and frame bits
(bits which are fully determined by the hash algorithm and the message size). In a tree of $f$-inputs, there are pointers from children to their corresponding parent. 
When a chaining value is present in a formatted $f$-input, it is
pointed by another $f$-input which is considered as its children. Each $f$-input has an associated index which locates it in the tree.
In addition, we need to define, for a tree $T$ of $f$-inputs, its corresponding \emph{tree template} $Z$ which has the same topology, where the corresponding $f$-inputs
have the same lengths and the frame bits match the corresponding bits in $T$, but where message and chaining value bits are not valuated. 
Thus, a tree template is fully determined by the message size 
and the parameters of the tree mode algorithm. This tree template
is used by the tree hash mode to instantiate the tree of $f$-inputs, by valuating progressively message bits and chaining value bits.
%all the inputs used during $f$-queries

A tree $T$ of $f$-inputs is said to be \emph{compliant} with a tree hash mode $\tau$ if the latter can produce a tree of $f$-inputs whose corresponding tree template
is compatible with it (its topology, the frame bits and the sizes of its $f$-inputs match those of $T$). A tree $T$ of $f$-inputs is said to be 
\emph{final-subtree-compliant} with $\tau$ if the latter can produce a tree of $f$-inputs whose proper subtree (\textit{i.e.} contaning at least the root $f$-input) 
has a corresponding tree template with which $T$ is compatible.

Bertoni \textit{et al.}~\cite{BDPV09,BDPV14_Suf}
give some guidelines to design correctly a tree hash mode $\tau$ operating an inner hash (or compression) function $f$. 
They define three sufficient conditions which ensure that the constructed hash function $\tau_f$, which makes use of 
an ideal hash (or compression) function $f$,
is indifferentiable from an ideal hash function. Besides, they propose to use particular frame bits in order to meet these conditions.
% % Pour que nous puissions concevoir correctement un mode de hachage en arbre opérant une fonction de hachage (ou de compression) $f$ dite ``interne'', 
% % Bertoni \textit{et al.} \cite{BDPV09,BDPV14_Suf} donnent 
% % des directives pour formater correctement les entrées de cette fonction. Ils définissent trois conditions suffisantes pour que la fonction
% % de hachage résultante, s'appuyant sur une fonction de hachage (ou de compression) idéale $f$, soit indifférenciable d'une fonction de hachage idéale.
% The formal definitions are quite complicated, and we prefer to give intuitive ideas for each of them:
% % Les définitions formelles de ces trois conditions étant 
% % %compliquées
% % complexes, nous préférons donner ici les idées intuitives de chacune d'elles~:
We refer to \cite{BDPV09,BDPV14_Suf} for the detailed definitions, and we give here a short description for each of them:
\begin{itemize}
 \item \textit{message-completeness:} 
 suppose we have a tree of $f$-inputs produced by the tree hash mode. There is an algorithm $A_{message}$ which, among the bits in the tree, 
 uniquely determines the message. 
 This requires that each message bit is processed at least once by $f$. The message can be reconstructed correctly if, given the sequence of bits of a $f$-input, 
 we can identify those which are message bits, and we are able to say what their positions in the message are. 
% without ambiguity in each $f$-inputs, and, when identied, we are able if we can identify the position of each
% message bits in $f$-inputs
%  Ayant toutes les entrées de $f$ utilisées lors de l'exécution du mode opératoire\index{mode opératoire}, nous pouvons déterminer 
%  de manière unique le message. 
 %In general, 
 Generally, only the end of the message is problematic. To cope with this, dedicated frame bits can be used such as a reversible 
 padding\footnote{Since a hash function processes an entire number of blocks (whose size $N$ depends on the underlying primitive), 
 a reversible padding is an efficient way of revealing the end of the message. This consists in applying to the message $M$, whatever its length, 
 a function $pad$ which returns a bit-string of length a multiple of $N$.
 Such a padding has to be reversible, \emph{i.e.} there is a function $unpad$ such that $unpad(pad(M))=M$ for all messages $M$.
 A well-known technique consists in appending the bit "$1$" to the end of the message, followed by the minimum number of bits "$0$", so that the total bit-length of the 
 padded message is a multiple of $N$.} for the message or a coding of the message length.
 The running time of $A_{message}$ should be linear in the total number of bits in the tree.
 \item \emph{tree-decodability:} 
 %Roughly speaking, the idea is to prevent making inner nodes out to be base level nodes. 
 %The tree of $f$-inputs can be parsed to retreive frame bits, chaining value bits and message bits unambiguously.
%  Given the set $S$ of $f$-inputs in a hash tree T generated by $H$, it is impossible to extract a 
%  subset $S' \subset S$ of values 
%  used to compute a proper subtree T' of T (containing at least the root $f$-input) which could have been generated 
%  legitimately by this hash function. 
 intuitively, given a tree T of $f$-inputs generated by $\tau$, it is impossible to extract a final proper subtree T' of T 
 %(\emph{i.e.} a proper subtree containing at least the root $f$-input) 
 which could have legitimately been generated by $\tau$. 
 %In other words, 
 In other words, given such a subset of $f$-inputs, we are able to say whether there is a missing input or not. More formally,
 this property is satisfied by $\tau$ if there are no trees of $f$-inputs which are both \emph{compliant} and \emph{final-subtree-compliant} with it, and there is a decoding
 algorithm $A_{decode}$ that can parse the tree progressively on subtrees, starting from the
root $f$-input, to retrieve frame bits, chaining value bits and message bits unambiguously. Also, when terminating, $A_{decode}$
must decide if the tree is \emph{compliant}, \emph{final-subtree-compliant}, or \emph{incompliant} with~$\tau$. 
The running time of $A_{decode}$ should be linear in the total number of bits in the tree.
 \item \textit{final-$f$-input separability:}  
 whatever the tree $T$ of $f$-inputs generated by $\tau$, we can distinguish between a root $f$-input and any other $f$-input.
%  Let us consider the set of hash trees constructed with the \emph{outer} hash function. For any of these hash trees,
%  we cannot find an inner node value
 Such a property is useful to prevent length extension attacks. 
 One straightforward way to fulfil this property is by means of domain separation between this final (root) $f$-input and other $f$-inputs, \emph{e.g.} 
 by augmenting them with a frame bit identifying them as such.
 %to use different compression functions to compress the root children and other leaf/inner nodes.
%  Nous pouvons faire la différence entre une entrée de $f$ ayant servie à obtenir un n{\oe}ud racine et une entrée 
%   ayant servie pour calculer n'importe quelle autre n{\oe}ud. Pour obtenir cette propriété qui évite les attaques par extension de longueur, un code 
%   peut être préfixé à l'entrée de la fonction $f$ pour indiquer s'il s'agit d'une entrée servant à calculer un n{\oe}ud racine ou un autre n{\oe}ud.
\end{itemize}
%\end{sloppypar}

% If the hash tree mode satisfies \emph{tree-decodability}, \emph{message-completeness} and \emph{final-node separability} then the resulting hash function
% is indifferentiable from a random oracle. 
These conditions ensure that no weaknesses are introduced on top of
the risk of collisions in the inner function. For instance, with \emph{tree-decodability}, an inner
 collision in the tree is impossible without a collision for the inner function.
Andreeva \textit{et al.} have shown in 
%\cite{AMP10,AMP10isc} 
\cite{AMP10,AMP10isc}
that a hash function indifferentiable from a random oracle
satisfies the usual security notions, up to a certain degree, such as pre-image and second pre-image resistance, 
collision resistance and multicollisions resistance.

\subsection{Definition of the inner hash function}

For our inner function, the hash functions based on the Merkle-Damg{\aa}rd construction, such as MD5, SHA-1 or SHA-2, have to be discarded. 
First, these functions cannot emulate a random oracle and we need this property to construct a tree-based hash function, 
constructed on top of it, which is still indifferentiable from a random oracle \cite{BDPV14_Suf}. Second, for efficiency purposes,
we want the inner function to have a running time linear in the number of blocks of the message. When MD-strengthened padding
is applied, the size of the message is appended to its end. This makes it difficult to obtain a running time perfectly linear in the number 
of blocks. As we will see, the \emph{prefix-free} Merkle-Damg{\aa}rd construction from Coron \emph{et al.} \cite{CDMP05} is a solution to both these problems.

Our inner hash function, based on the \emph{prefix-free} MD construction, iterates a SBL (Single-Block-Length) compression function, 
\emph{i.e.} whose output length is the same as the
message block length. We use the compression function
$c: \{0,1\}^{2N} \rightarrow \{0,1\}^N$ based on a 
$N$-bit block cipher with a $N$-bit key, such as the Davies-Meyer compression function: $c(x,y)=E_y(x) \oplus x$ where $x$ 
is the previous hash state and $y$ is a block of the message. Such a hash function \cite{CDMP05}, denoted $f$, consists in applying the 
plain MD construction to a prefix-free encoding of the message input.~\\

\noindent\fbox{%
\begin{varwidth}{\dimexpr\linewidth-2\fboxsep-2\fboxrule\relax}
\paragraph{{\boldmath\textbf{The considered inner function $f$}}}~\\
INPUT: message $m$.~\\
OUTPUT: hash value.
\begin{enumerate}
 \item The message $m$ is padded with $10^r$ where $r$ is the minimum number (possibly zero) of bits $0$ such that its bitlength 
 is a multiple of $N$.
 \item The $N$-bits encoding of the number of blocks is prepended to the message (prefix-free encoding step).
 \item The message is parsed into $k$ blocks $m_1$, $m_2$,..., $m_k$ of size $N$ bits.
 \item The plain MD mode is applied on these $k$ blocks.
 \begin{itemize}
  \item[] Let $y_0 = 0^N$ (or a fixed IV value).
  \item[] For $i=1$ to $k$ do $y_i \gets c(y_{i-1},m_i)$.
  \item[] Return $y_k \in \{0,1\}^{N}$.
 \end{itemize}
\end{enumerate}
\end{varwidth}
}~\\

At first sight, due to the padding and the prefix-free encoding, this hash function requires $j+2$ calls to the underlying compression function 
to process a message of $jN$~bits. In fact, the node arities of our tree topologies can be upper-bounded by a constant. Thus, the first block of
the prefix-free encoding of an input involved in our trees has a constant number of possible values, and their possible corresponding hash states 
$y_1$ can be precomputed. Assuming a constant number of precomputed hash states, the running time of this hash function is then reduced to $j+1$
calls to the underlying primitive. Hence, it is explained in the subsection above that before applying the function~$f$, domain
separation bits have to be added to the input. Bertoni \emph{et al.} \cite{BDPV14_Suf} have stated that 2 bits is sufficient. 
The number of distinct 
domain separation codes can then be considered small.
For domain separation purposes, we choose to prepend a code of $N-1$ bits to the message 
so that there is no extra bits $0$ due to the padding operation. With this second ``large'' encoding, the first bit of 
the message is at the end of the second block. Then, the running time is still $j+1$ when the possible values of
$y_1$ are precomputed. Since the number of distinct domain separation codes is small and that only one bit of the message is in the second block, 
we can precompute the possible values of $y_2$ (instead of $y_1$)
so as to reduce the running time of $f$ to $j$ units of time. Further details will be provided in Section \ref{sec_section2}.
~\\

A hash function that performs one call to a SBL compression function to process one block of the (padded) message is also called a SBL hash function.~\\

% \paragraph{About the use of other SBL compression functions}
% Chang \emph{et al.} \cite{CLNY06} have extended the work of Coron \emph{et al.} by checking if other 
% hash functions using this prefix-free MD construction are still indifferentiable from an RO.
% It turns out that sixteen of the PGV (Preneel Govaerts Vandewalle \cite{PGV93}) compression functions yield sound hash functions.
% Again, these hash functions divide the message in blocks of the same size than the digest.
% ~\\
% 
% \paragraph{About the padding overhead in other designs} The UBI transformation function of Skein~\cite{FLSWBKCW09} is collision resistant and 
% requires exactly $j$ calls to the tweakable
% block-cipher Threefish to process a message of $j$ blocks. Indeed, when the message size is already a multiple of $N$, a flag in the tweak indicates whether
% or not the message is padded. Such a design choice is also present in the BLAKE2~\cite{ANWW13} hash function.

\paragraph{On the use of other SBL compression functions}
Chang \emph{et al.} \cite{CLNY06} have extended the work of Coron \emph{et al.} by checking if other 
hash functions using this prefix-free MD construction are still indifferentiable from an RO.
It turns out that sixteen of the PGV (Preneel Govaerts Vandewalle \cite{PGV93}) compression functions yield sound hash functions.
Again, these hash functions divide the message in blocks of the same size than the digest.~\\

\paragraph{On the use of a sponge-based hash function as inner function}
The function $f$ can be a sponge-based hash function~\cite{BerDaePeeVan08} if the message block and output sizes are equal.
This constraint can be fulfilled by setting the rate (also denoted $N$) of this sponge function to the output size. Let us suppose that we use
a function having the same padding rule than Keccak~\cite{BerDaePeeVan13}. This padding rule consists in adding to the end of the message 
the bitstring $10^*1$ where $0^*$ is
the minimum number of bits $0$ such that the padded message has a size a multiple of $N$. 
Hence, we choose to use $N-2$ bits for the prepended domain separation codes so that the padded message ends with the two bits~$11$.
With this large encoding, the first two bits of the message are at the end of the first block. Taking into account these overheads, the running time of
$f$ is $j+1$ to process a message of $j$ blocks.
Some precomputations can be done to reduce the running time of $f$ to $j$ units of time. 
Further details will be provided in Section \ref{sec_section2}.~\\

\paragraph{About the padding overhead in other designs} The UBI transformation function of Skein~\cite{FLSWBKCW09} is collision resistant and 
requires exactly $j$ calls to the tweakable
block-cipher Threefish to process a message of $j$ blocks. Indeed, when the message size is already a multiple of $N$, a flag in the tweak indicates whether
or not the message is padded. Such a design choice is also present in the BLAKE2~\cite{ANWW13} hash function.

\subsection{Parallel computation model}

We make the assumption that the number of processors is equal to the number of nodes at the first level of the tree. Once the nodes have been computed 
at a given level, the processors are reused to process the next (upper) level.

We use the classic PRAM (Parallel Random Access Machine) model of computation~\cite{GiRy88},
assuming the strategy EREW (Exclusive Read Exclusive Write). When dealing with hash trees, this model can indeed 
be restricted to this strategy: we do not need that two processors write simultaneously into the same memory location, nor that a same data
can be read simultaneously by two or more processors. In the context of parallel hashing, 
it serves \textit{a priori} no purpose to process twice a same message block or chaining value.
%that concurrent reads of a same message block or chaining value never occur.

Let us denote by $L_i$ the list of nodes at level $i$. 
Given the definitions of a level arity and of our inner hash function, 
the parallel running time to process a hash tree of height $h$ is equal to 
\[\sum_{i=1}^h \max_{\mathrm{Node} \in L_i} \mathrm{arity}(\mathrm{Node}),\]
where the function $\mathrm{arity()}$ returns the arity of a node.
The \textit{total work} to process a hash tree is equal to
\[\sum_{i=1}^h \sum_{\mathrm{Node} \in L_i} \mathrm{arity}(\mathrm{Node}).\]
In other words, this quantity corresponds to the running time when it is executed 
sequentially (\textit{i.e.} by a single processor).

\section{Optimization of hash trees for parallel computing}\label{sec:optim}

\subsection{Minimizing the running time}\label{subsec:runni}

In order to optimize the running time of a tree mode, we make a certain degree of flexibility 
%possible regarding the 
on the choices of node arities. 
We can note immediately that allowing different %nodes arities 
node arities in a same level of the tree provides no efficiency gains. Worse, 
the running time may be less interesting since a tree level processing running time is bounded by the running time to process the node 
having the highest arity. 
This observation suggests that, in order to hope a reduction of the tree processing running time, %nodes' arities 
node arities at the same level need to be set 
to the same value\footnote{Except maybe the rightmost node which may be of smaller arity.} while allowing arities to vary from one level to another.
%Our strategy will therefore consist in allowing a different arity at each level of the tree.
Therefore our strategy allows a different arity at each level of the tree.

Let us denote $l$ the block-length of a message. The problem is to find a tree height $h$ and integer arities $x_1$, $x_2$, ..., $x_h$ 
such that $\sum_{i=1}^h x_i$ is minimized. When constructing a hash tree having its leaves at the same depth,
we seek integers $x_1$, $x_2$, ..., $x_h$ such that $\lceil \lceil \cdots \lceil \lceil l/x_1 \rceil / x_2 \rceil \cdots \rceil/ x_h \rceil = 1$.
Since we have 
\[\lceil \lceil \cdots \lceil \lceil l/x_1 \rceil / x_2 \rceil \cdots \rceil/ x_h \rceil = \lceil l/(x_1x_2 \cdots x_h) \rceil\]
for (strictly) positive integers $(x_j)_{j=1 \ldots h}$,
%Supposons que les $x_i$ soient notées dans l'ordre décroissant des arités. 
any solution to the problem must necessarily satisfy the following constraints:
\begin{equation}\label{contraintes}
\prod_{i=1}^h x_i \ge l\ \textrm{and}\ 
\left(\prod_{i=1}^{h} x_i\right)/x_j < l\quad \forall\ j \in \llbracket 1,h \rrbracket.
\end{equation}

Note that if a solution $(x_1, x_2 , \ldots, x_h)$ does not satisfy the second contraint in~(\ref{contraintes}), this means that a better solution
exists. A solution to this problem is a multiset of arities.
% We can wonder what is the best solution between a binary and a ternary tree.
%Nous pouvons d'abord nous poser la question de quelle est le meilleur choix entre un arbre binaire et un arbre ternaire. 
First, we show that, in a non-asymptotic setting, 
%asymptotically regarding the size of the message, 
a perfect ternary tree 
%approaches 
comes closer to optimality 
than a perfect binary tree.
Then we examine the case of trees having different arities at each level.

% Afin d'avoir une idée, on peut commencer par considérer que les $x_i$ et $h$ sont des 
% nombres réels. Ainsi il convient de minimiser la somme des $x_i$ sachant que 
% leur produit est $l$. On sait que le minimum est atteint lorsque les $x_i$
% sont égaux à un même nombre qu'on notera $x$. Donc:
% $x^h=l$ c'est à dire $x=l^{\frac{1}{h}}$. Il faut maintenant
% déterminer $h$ de telle sorte que $h l^{\frac{1}{h}}$ soit minimum.
% Un calcul de dérivée montre que ce minimum est atteint pour $h=\ln(l)$,
% ce qui implique $x=e$. 
% We can wonder what is the best solution between a perfect binary tree and a perfect ternary tree.
% Ceci nous incite à étudier les deux cas  particuliers, celui 
% des arbres parfaits où toutes les arités sont $2$
% et celui des arbres parfaits où toutes les arités sont $3$. 
% La comparaison de ces deux cas est faite dans l'annexe \ref{comp} et montre
% qu'à partir d'une certaine taille pour $l$ l'arbre ternaire parfait donne 
% un meilleur temps d'exécution que l'arbre binaire parfait.

First of all, we can start by considering the $h$ and $x_i$ (for $i=1 \ldots h$) as real numbers. Thus, 
we have to minimize the summation of $x_i$ subject to the constraint that their product is $l$. 
We know that the minimum is reached when the $x_i$
are equal to the same number, which we will denote $x$. So we have
$x^h=l$, that is $x=l^{\frac{1}{h}}$. We must now determine
$h$ so that $h l^{\frac{1}{h}}$ is minimized.
The calculation of a derivative shows that this minimum is reached for $h=\ln(l)$,
which implies $x=e$. 
Consequently, we can wonder what the best solution is between a perfect binary tree and a perfect ternary tree.
% Ceci nous incite à étudier les deux cas  particuliers, celui 
% des arbres parfaits où toutes les arités sont $2$
% et celui des arbres parfaits où toutes les arités sont $3$. 
The comparison of these two cases is done in Appendix \ref{comp} and shows that
beyond a certain message size $l$ ($l=2^{28}$), a perfect ternary tree gives a better running time than a perfect binary tree.
% qu'à partir d'une certaine taille pour $l$ l'arbre ternaire parfait donne 
% un meilleur temps d'exécution que l'arbre binaire parfait.
In fact, as the present general study shows, a tree having different level arities can give better results.
Let us remind that node arities are not allowed to vary in a same level (same stage) of the tree.
A level of the tree is said to be of arity $a$ when all nodes at this level are of arity at most~$a$.
Given an optimal tree (in the sense of the running time) for hashing, we can ask what the possible arities are for its levels. 
%On peut d'abord se poser la question des arités possibles que peuvent admettre les niveaux d'un arbre optimal en temps d'exécution. 
We have the following Theorem:

\begin{theorem}\label{arites}
For a hash tree whose running time is optimal, the following statements hold: 
\begin{itemize}
 \item It can be comprised of levels of arity $2$, $3$, $4$, or $5$. Higher arities are not possible.
 \item It can be constructed using only levels of arity $2$ and $3$.
\end{itemize}
\end{theorem}

\begin{proof}
We prove these two assertions separately:
\begin{itemize}
 \item We first show that levels of arity $a$ with $a \geq 7$ lead to trees having a suboptimal running time.
Indeed, any node of arity $a \geq 7$ can be replaced by a tree of arity $2$ having a better running time. 
We simply have to note that $2\lceil \log_2 a \rceil < a$ for all $a>6$, meaning that
a $a$-ary tree of height $1$ can be advantageously replaced by a binary tree of height $\lceil \log_2 a \rceil$. 
In contrast, for all nodes of arity $a$ with $a \in \llbracket 3, 6 \rrbracket$ and for all $i \in \llbracket 2,5 \rrbracket$ 
we have $i\lceil \log_i a \rceil \geq a$. 
Finally, a node of arity $6$ can be replaced by a $\{3,2\}$-aries tree, since $2 \cdot 3 = 6$,
thereby reducing the running time to $2+3 = 5$ units. 
 \item As regards the second assertion, a node of arity $5$ can be replaced by a tree of arities $\{3, 2\}$, since $2 \cdot 3 = 6 > 5$. 
This transformation does not
change the running time since $2+3 = 5$. Finally, a node of arity $4$ can be replaced by a binary tree of height $2$ for a running time
which is still unchanged.
\end{itemize}
\end{proof}

An optimal tree has not necessarily a single topology. Firstly, a solution satisfying constraints
(\ref{contraintes}) can be defined as a multiset of arities since we can permute them.
For instance, suppose a tree has three levels with the first level of arity $3$, the second one of arity $2$ and the last one (that is, the root node) 
of arity $3$. We can permute these arities so that the first level is of arity $2$ and the latter two levels of arity $3$. 
If this new tree has the same running time, its topology has however changed.
Secondly, we can find examples where different multisets of arities lead to trees of optimal running time. For instance, if we consider a $7$-block message, 
the multisets of arities $\{2, 2, 2\}$, $\{3, 3\}$ and $\{4, 2\}$ allow the construction of trees having the optimal running time (see 
Figure~\ref{Same_running_time_different_trees}).
We can, however, construct optimal trees by restricting the set of possible arities. We have the following theorem:
%mais nous avons le thér
% \begin{theorem}
% A hash tree whose running time is optimal can be constructed using levels of arity $2$ and $3$.
% \end{theorem}

\begin{figure}[t!]
\centering
%\begin{subfigure}[t]{.4\textwidth}
\subfloat[$(2, 2, 2)$-aries tree]
  {
  \centering
  \includegraphics[scale=0.6]{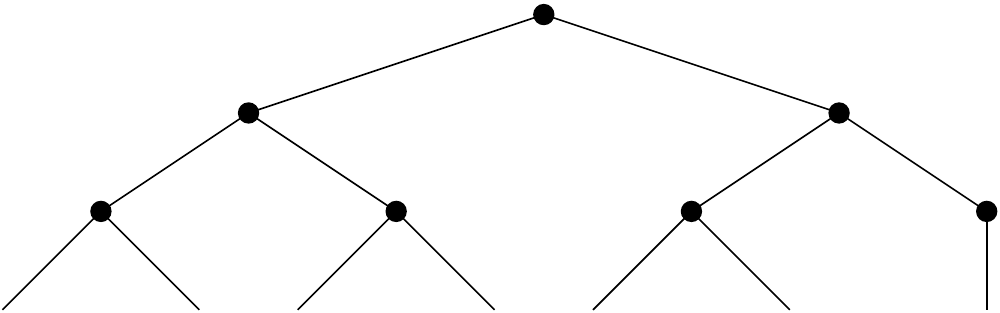}
  %\caption{Running time of an optimal tree (shown in blue) compared to a binary tree (in black)}
  \label{perf1}
  }
%\end{subfigure}%
~~~~~~~~~~
%\begin{subfigure}[t]{.4\textwidth}
\subfloat[$(3, 3)$-aries tree]
  {
  \centering
  \includegraphics[scale=0.6]{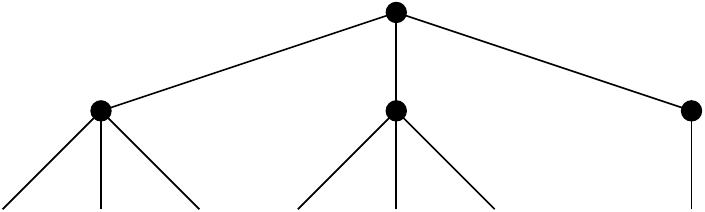}
  %\caption{Running time of an optimal tree (shown in blue) compared to a binary tree (in black)}
  \label{perf1}
  }
%\end{subfigure}%

\vspace{0.75cm}
%\begin{subfigure}[t]{.4\textwidth}
\subfloat[$(4, 2)$-aries tree]
  {
  \centering
  \includegraphics[scale=0.6]{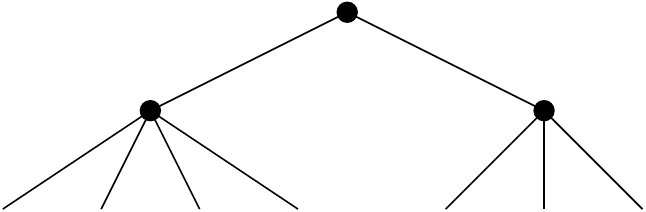}
  %\caption{Speed gain of an optimal tree compared to a binary tree}
  \label{perf2}
  }
%\end{subfigure}
\caption{Different topologies for a $7$-block message}
\label{Same_running_time_different_trees}
\end{figure}

\begin{theorem}\label{min_running_time_3_cases}
Let a message of length $l$ blocks and let $i$ be the lowest integer such that $3^i \geq l$. Let us note $x \in \llbracket 0,2 \rrbracket$ the value which minimizes
the product $3^{i-x}2^x$ under the constraint $3^{i-x}2^x \geq l$.  There exists an optimal tree (in the sense of optimal running time) which has $i-x$ levels 
of arity $3$ and $x$ levels of arity $2$. More precisely, we can state the followings:
\begin{itemize}
 \item If $l \leq 3^i < \frac{3l}{2}$ then a ternary hash tree is optimal for a running time of $3i$.
 \item If $\frac{3l}{2} \leq 3^i < \frac{9l}{4}$ then an optimal hash tree has $i-1$ levels of arity $3$ and one level of arity $2$,
for a running time of $2+3(i-1)$.
 \item Otherwise $\frac{9l}{4} \leq 3^i < 3l$, and then an optimal hash tree has $i-2$ levels of arity $3$ and $2$ levels of
arity $2$, for a running time of $4+3(i-2)$.
\end{itemize}
Such an optimal tree maximizes the number of levels of arity $3$.
\end{theorem}

\begin{proof}
According to Theorem \ref{arites}, a hash tree whose running time is optimal can be constructed using only levels of arity $2$ and levels of arity $3$.
We still need to find out their numbers.
If we have at least $3$ levels of arity $2$ then we can replace these $3$ levels by $2$ levels of arity $3$ ($3^2=9>2^3=8$).
The running time to process $3$ levels of arity $2$ or $2$ levels of arity $3$ is $6$. 
Therefore, it is always possible to construct optimal trees with maximum $2$ levels of arity $2$. 
Let $i$ be such that $3^i \geq l$.
From the parallel running time standpoint, it is preferable to trade a level of arity $3$ for a level of arity $2$.
This means that the sought solution corresponds to the highest value $x \in \llbracket 0,2 \rrbracket$ 
 such that $3^{i-x}2^x\geq l$. The three assertions follow immediately.
\end{proof}

% \begin{remark}
%  Let $i$ be such that $3^i \geq l$. Its is not difficult to see that the sought solution corresponds to the highest value $x \in \llbracket 0,2 \rrbracket$ 
%  such that $3^{i-x}2^x\geq l$.
% \end{remark}

~\\
To determine the level arities of an optimal tree, we apply the following algorithm:~\\

% %\begin{sloppypar}
% \paragraph{\textbf{Algorithm 1}} 
% This algorithm takes as input a message length $l$
% and returns a multiset of arities
% minimizing the running time.
% We first compute $i~=~\lceil \log l/\log 3 \rceil$ 
% and then $x=\left\lfloor \log(l/3^i) /\log(2/3) \right\rfloor$. 
% The returned multiset consists of $i-x$ first levels of arity $3$ and $x$ last levels of arity $2$.
% %The $i-x$ first levels are of arity $3$ and the last $x$ levels of arity $2$.
% %floor(log(x/(3^(ceiling(log(x)/log(3)))))/log(2/3))
% %\end{sloppypar}

\noindent\fbox{%
\begin{varwidth}{\dimexpr\linewidth-2\fboxsep-2\fboxrule\relax}
%\begin{sloppypar}
\paragraph{\textbf{Algorithm 1}} ~\\
INPUT: a message length $l$.~\\
OUTPUT: a multiset of arities minimizing the running time.
\begin{itemize}
 \item We first compute $i~=~\lceil \log l/\log 3 \rceil$ 
and then $x=\left\lfloor \log(l/3^i) /\log(2/3) \right\rfloor$. 
 \item We return a multiset which consists of $i-x$ first levels of arity $3$ and $x$ last levels of arity $2$.
\end{itemize}
%The $i-x$ first levels are of arity $3$ and the last $x$ levels of arity $2$.
%floor(log(x/(3^(ceiling(log(x)/log(3)))))/log(2/3))
%\end{sloppypar}
\end{varwidth}% 
}

~\\
\paragraph{\textbf{Examples}} For messages of lengths $l=4, 5$ and $10$ blocks respectively, Algorithm 1 returns the multisets 
of arities $\{2,2\}$, $\{3,2\}$ and $\{3, 2,2\}$ respectively.
The number of processors is not optimized here. This aspect is addressed in the following section. ~\\

% \noindent
% \textbf{Why is it possible to minimize the running time with a tree whose leaves are at the same depth?}
% Let us suppose that we have, for a given message length, an optimal tree whose leaves are not at the same depth.
% Then, for each leaf located at a level greater than zero, we can create descendants in order to complete the tree so that all leaves
% are at level $0$. It is possible to perform this while keeping a tree of same height and respecting the level arities.

The result can be either a \emph{perfect} tree where the number of leaves is greater than the message length (the tree is said to be \emph{perfect} since, on the one hand, 
nodes at a same level are all of same arity, and, on the other hand, all the leaves are at the same depth), or a \emph{truncated} tree since it 
is possible to prune some right branches to remove this surplus of leaves.
% Consequently, there exists necessarily a tree
% having the same height, the same multiset of arities and a lower number of leaves corresponding to the message length.
In the rest of the paper, we refer to a truncated $(x_1,x_2, \ldots, x_h)$-aries tree to speak about a tree having a number of leaves equal 
to the message length and where the nodes of the base level are of arity at most $x_1$, nodes at the second level are of arity at most $x_2$ and so on.~\\

As a last remark, since the hash function must be deterministic, the multiset of arities must also be chosen deterministically 
as a function of the message size. For instance, we can arrange in descending order the elements of the multiset of arities. 
The solution to the problem of minimizing the running time is then uniquely determined as an ordered multiset.

\paragraph{\textbf{Performance improvements}}
We have seen that for a message of 6 blocks (see Figure~\ref{Arbre_exemple}), the performance gain of an optimal tree compared to a binary tree is 20\%.
Figure~\ref{perf1} shows the running times of an optimal tree and a binary tree as functions of the message size varying from $1$ to $10^5$ blocks. 
Figure~\ref{perf2} shows the speed gain obtained with an optimal tree. 
The gain in time (or speedup gain) 
is computed as $100(T_b/T_o~-~1)$ where $T_b$ is the running time of a binary tree and $T_o$ the running time of an optimal tree. 
As we can see, the gain differs from one message size to another. The gain can be greater than 30\% for very short messages
but decreases quickly, to cap at 10\%. 
As regards the message size, although the diagram
does not cover a sufficiently long range, one can note a slight downward slope. 
%This shows that it is increasingly difficult to get close to the optimal average arity by using integer arities. 
%\vspace{-0.75cm}
\begin{figure}[htb]
\centering
%\begin{subfigure}[t]{.4\textwidth}
\subfloat[Running time of an optimal tree (shown in blue) compared to a binary tree (in black)]
  {
  \centering
  \includegraphics[scale=0.31]{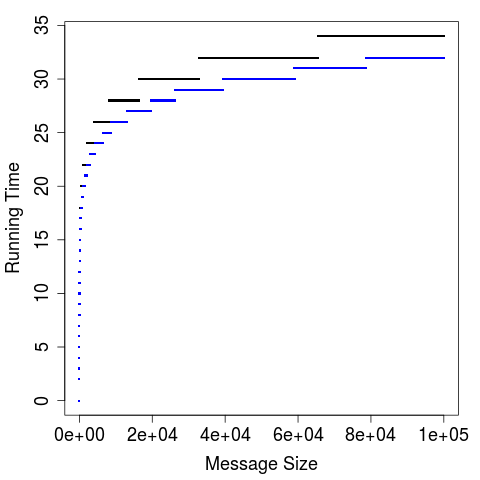}
  %\caption{Running time of an optimal tree (shown in blue) compared to a binary tree (in black)}
  \label{perf1}
  }
%\end{subfigure}%
%~~~~~~~~~~~~~~~~~
%\begin{subfigure}[t]{.4\textwidth}

\subfloat[Speed gain of an optimal tree compared to a binary tree]
  {
  \centering
  \includegraphics[scale=0.31]{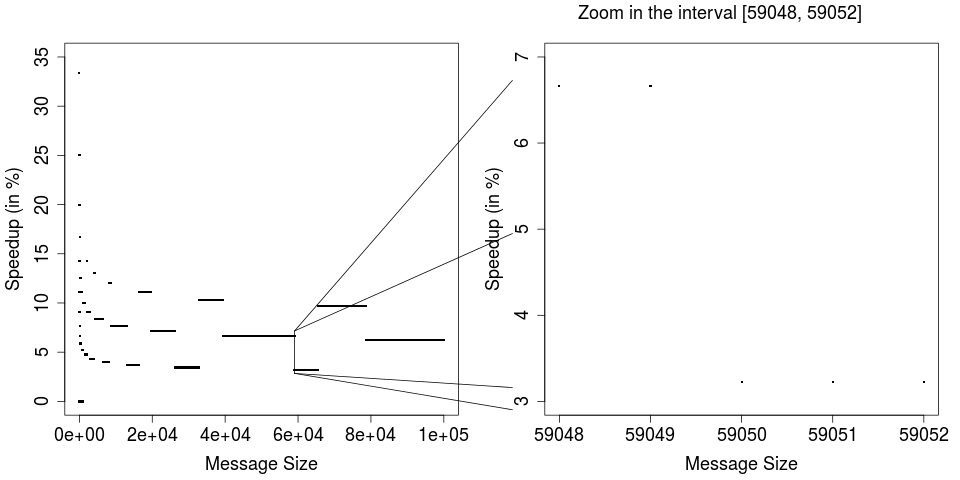}
  %\caption{Speed gain of an optimal tree compared to a binary tree}
  \label{perf2}
  }
%\end{subfigure}
\caption{Performance comparison between an optimal tree and a binary tree}
\label{Gains_execution}
\end{figure}

%\vspace{-0.3cm}
\subsection{Minimizing the number of processors}\label{subsec:proce}

In this section we look into how to minimize the number of required processors to obtain the optimal running time.
We have two cases to study, the trees having all leaves at the same depth and the others. 
We fully treat the first case and we make a few observations regarding the second type of tree, 
%which intuitively should further reduce
which we intuitively sense to further reduce the number of required processors.~\\

At the outset, one may be interested in the maximum possible number of levels of arity $5$ or $4$. We have the following Lemma:

\begin{lemma}\label{numb_ar_4_and_5}
In a tree having an optimal running time the following statements hold:
\begin{itemize}
 \item There can be up to $1$ level of arity $5$; 
 \item There can be up to $6$ levels of arity $4$.
\end{itemize}
\end{lemma}

\begin{sloppypar}
\begin{proof}
We prove these two assertions separately:
\begin{itemize}
 \item Suppose that the tree has $2$ levels of arity $5$. We 
%can therefore 
replace these $2$ levels by $3$ levels of arity $3$ since $3^3=27>5^2=25$.
The running time is improved since $3\cdot3~=~9<2\cdot5=10$. 
We can then state that $2$ levels of arity $5$ lead to a tree having a sub-optimal running time. 
 \item Now, let us look for a pair of minimum integers $(i,j)$ satisfying $3^i \geq 4^j$ and $3\cdot i<4\cdot j$. 
The first pair which satisfies these constraints is $i=9$ and $j=7$. We can then replace
$7$ levels of arity $4$ by $9$ levels of arity $3$ in order to decrease the running time. 
\end{itemize}
\end{proof}
\end{sloppypar}
%\subsubsection{Trees having all leaves at the same level.}

\begin{sloppypar}
We have seen that it is possible to construct a tree optimizing the running time by using only levels of arity $2$ and $3$.
% de construire un arbre optimal constitué 
% soit d'un seul niveau d'arité $2$ et les autres d'arités $3$, 
% soit de $2$ niveaux d'arités $2$ et les autres d'arités $3$, soit uniquement de niveaux d'arités $3$. 
%de niveaux d'arité $2$ et $3$.
In what follows, we 
%are going to 
show how 
%we can 
to deduce an optimal tree minimizing the number of involved processors.
Let us suppose that level arities $x_1$, $x_2$ ..., $x_h$ are noted in (no strictly) decreasing order so that $x_1$ 
is the arity of the base level and $x_h$ the arity of the last level, 
\textit{i.e.} the arity of the root node. 
The trees optimizing the running time, defined above, are not necessarily full in the sense that a rightmost node at a given level can be of arity 
strictly lower than the arity of this level. First, we note that for the trees constructed with Algorithm 1, the number of required processors
is equal to $\lceil l/3 \rceil$ in the best case, and equal to $\lceil l/2 \rceil$ when there are only levels of arity $2$. Moreover, 
according to Theorem \ref{arites} 
we know that a level arity cannot be greater than $5$. This means that in the best case, after optimization, the number of required processors 
could be reduced to $\lceil l/5 \rceil$. Thus, we could in the best case decrease the number of processors 
by a factor of about $5/2$.
\end{sloppypar}

Given an optimal tree for the running time, the intent is to increase the arity of the first level (base level) while decreasing arities 
of the following levels so that the sum of the level arities remains constant and their product remains greater than or equal to $l$.
To solve this problem we propose in Appendix \ref{red_number_processors} two solutions (Algorithm 2a or 2b). However, as will be discussed below,
we can further optimize hash trees.

According to Theorem \ref{arites}, a level arity of a tree minimizing the running time cannot exceed $5$. 
Thus, Algorithm 2a (or Algorithm 2b) of Appendix \ref{red_number_processors} allows us to substitute any sub-multiset $A$ for 
another one, denoted $A'$, where the sum of arities remains the same, 
and by trying to increase the arity of the base level up to $5$.
% L'algorithme 2 permet alors de substituer n'importe 
% quel sous-multiensemble d'arités $A$ par un autre, noté $A'$, dont la somme des arités reste la même, 
% et en essayant d'augmenter sa première arité jusqu'à $5$.
Consider, for instance, a message of size $l=95$ blocks. With such a message size, Algorithm 1 returns the multiset of arities $A_0=\{3,3,3,2,2\}$
which defines a tree structure involving $32$ processors.
% Prenons par exemple un message de taille $l=5$ blocs dont l'algorithme 1 nous retourne le multi-ensemble d'arité $A_0=\{3,2\}$. 
By applying Algorithm 2, we obtain the multiset $A_1=\{4,3,3,3\}$ which reduces the number of involved processors to $24$ 
while leaving the running time unchanged.
% En appliquant l'algorithme 2, l'étape 1
% nous donne le multiensemble $A_1=\{5\}$ qui réduit le nombre de processeurs impliqués à $1$ en laissant le temps d'exécution inchangé.
%%%%%%%%%%%%%%%%%%%
% % % Considering message sizes varying from $4$ to $16 \cdot 10^6$ blocks for which we have optimized the corresponding multisets in time and width, 
% % % we can study the distribution of maximum arities for each of them. 
% % % Having runned Algorithm 2 on this range of sizes, the maximum arity appears to be: 3 in 24.9\% of cases, 4 in 35.2\% of cases and 5
% % % in 39.9\% of cases.
%%%%%%%%%%%%%%%%%%%
% Si nous nous intéressons à l'arité maximale que contient un multi-ensemble issu de l'algorithme 2 pour une taille de message variant 
% de $4$ à $16\cdot10^6$ blocs, elle est de 3 dans 24,9\% des cas, de 4 dans 35,2 \% des cas et de 5 dans 39,9 \% des cas.
~\\

\noindent
% \paragraph{\textbf{What if we were to increase the arity of each level ? As much as possible ?}} 
\textbf{What if the arity of each level is increased? As much as possible?}
We just saw 
%how 
that we can increase the arity of the first level.
It would also be preferable to increase the arity of each level of the tree in order to free up the highest number of processors 
at each step of the computation. An example is depicted in Figure \ref{Gain_proc}.
%Nous venons de voir comment augmenter l'arité du premier niveau. Il y a également un
%intérêt à augmenter l'arité de chaque niveau, autant que l'on peut, celui de libérer à chaque étape du calcul le plus grand nombre de processeurs possible.
\begin{figure}[h]
  \centering
  \resizebox{1\textwidth}{!}{
  \includegraphics{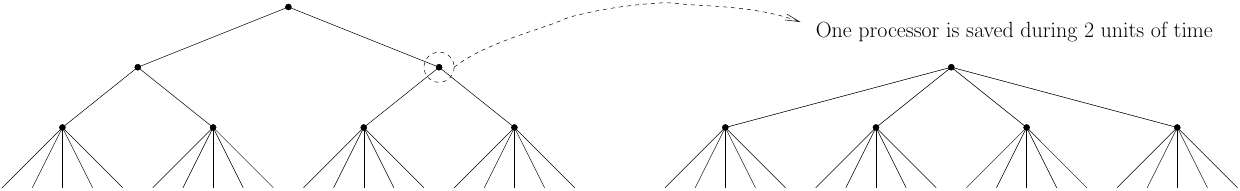}
 }
 \caption{Two trees compressing a 20-block message, optimized both for the running time and the number of involved processors. 
 Both trees require 4 processors.
 Nevertheless, we note that the right tree is the best choice. Indeed, the one on the left needs 4 processors during 5 units of time,
 then 2 processors during 2 units of time, and finally one processor during 2 units of time. The one on the right needs 4 processors during 5 units of time and then one processor during 4 units 
 of time.}
\label{Gain_proc}
\end{figure}

While we propose an iterative algorithm in Appendix \ref{red_number_processors} to construct an optimal tree maximizing the arity of each level, we also 
enumerate all possible cases in the following Theorem:

\begin{theorem}\label{min_running_time_max_arities}
For any integer $l \geq 32$ there is an unique ordered multiset $A$
of $h_5$ arities $5$, $h_4$ arities $4$, $h_3$ arities $3$ and $h_2$ arities $2$
such that the corresponding tree covers a message size $l$, has a minimal running time
and has first $h_5$ as large as possible, then $h_4$ as large as possible, and then 
$h_3$ as large as possible.
More precisely, if $i$ is the lowest integer such that $l \leq 3^i < 3l$, 
this ordered multiset is such that:
\[
\left\{
\begin{array}{llrll}
|A|=i,    h_5=0,  h_4=0,  h_3=i,    h_2=0  & \hbox{ if } & l \leq & 3^i & < \frac{9l}{8}, \\
|A|=i,    h_5=0,  h_4=1,  h_3=i-2,  h_2=1 & \hbox{ if } & \frac{9l}{8} \leq & 3^i & < \frac{81l}{64}, \\
|A|=i-1,  h_5=0,  h_4=3,  h_3=i-4,  h_2=0 & \hbox{ if } & \frac{81l}{64} \leq & 3^i & < \frac{27l}{20}, \\
|A|=i-1,  h_5=1,  h_4=1,  h_3=i-3,  h_2=0 & \hbox{ if } & \frac{27l}{20} \leq & 3^i & < \frac{3l}{2}, \\
|A|=i,    h_5=0,  h_4=0,  h_3=i-1,  h_2=1 & \hbox{ if } & \frac{3l}{2} \leq & 3^i & < \frac{27l}{16}, \\
|A|=i-1,  h_5=0,  h_4=2,  h_3=i-3,  h_2=0 & \hbox{ if } & \frac{27l}{16} \leq & 3^i & < \frac{9l}{5}, \\
|A|=i-1,  h_5=1,  h_4=0,  h_3=i-2,  h_2=0 & \hbox{ if } & \frac{9l}{5} \leq & 3^i & < \frac{81l}{40}, \\
|A|=i-1,  h_5=1,  h_4=1,  h_3=i-4,  h_2=1 & \hbox{ if } & \frac{81l}{40} \leq & 3^i & < \frac{9l}{4}, \\
|A|=i-1,  h_5=0,  h_4=1,  h_3=i-2,  h_2=0 & \hbox{ if } & \frac{9l}{4} \leq & 3^i & < \frac{81l}{32}, \\
|A|=i-1,  h_5=0,  h_4=2,  h_3=i-4,  h_2=1 & \hbox{ if } & \frac{81l}{32} \leq & 3^i & < \frac{27l}{10}, \\
|A|=i-1,  h_5=1,  h_4=0,  h_3=i-3,  h_2=1 & \hbox{ if } & \frac{27l}{10} \leq & 3^i & < 3l, \\
\end{array}
\right .
\]
where the number $h_3$ is at least $1$ in the first case and can be $0$ in the other cases. 
\end{theorem}

\begin{proof}
Let us start from the 3 cases of Theorem \ref{min_running_time_3_cases} which maximize the number of levels of arity 3. 
For a given message length $l$, we consider the corresponding optimal tree (in the sense of the running time). 
We denote by $a$ the initial number of levels of arity $2$ and by $i-a$ the initial (maximized) number of levels of arity $3$. 
We want to transform this tree in order to increase the arity of each level as much as possible, while leaving the running time unchanged.
According to Lemma \ref{numb_ar_4_and_5}, there can be one level of arity 5 and up to six levels of arity 4.
Since we want to maximize the number of levels of arity 4 after having maximized the number of levels of arity 5, 
there cannot be more than one level of arity 2. Thus, $h_5 \in \llbracket 0, 1 \rrbracket$, $h_4 \in \llbracket 0, 6 \rrbracket$ 
and $h_2 \in \llbracket 0, 1 \rrbracket$, meaning there shall be at most 28 cases.
Note that among these 28 cases, many may not be valid solutions.
The aim is to transform the initial product $2^a3^{i-a}$ into
a product $2^w3^{i-a-b}4^v5^u$ where $a$ and $b$ are respectively the number of levels of arity $2$ and the number of levels of arity $3$ 
that we have transformed, and $u$, $v$, $w$ the number of levels of arity $5$,
$4$, $2$ respectively.
For each triple $(h_5=u,h_4=v,h_2=w)$ with $u \in \llbracket 0, 1 \rrbracket$, $v \in \llbracket 0, 6 \rrbracket$ 
and $w \in \llbracket 0, 1 \rrbracket$, we can verify that there is a solution $(a,b)$ with $a$ an integer in $\llbracket 0,2 \rrbracket$
and $b$ a positive integer such that $3b + 2a = 5u + 4v + 2w$. We remark that $3b + 2a$ can be rewritten $3b$ if $a=0$, $3(b+1)+1$ if $a=2$, and $3b+2$ if $a=1$. 
Thus, all but one integers are produced.
We can also verify that this solution is unique. Indeed, let us suppose a second solution $(a',b')$. Since $3b+2a=3b'+2a'$, we have
$3(b'-b)=2(a'-a)$, meaning that $3$ divides $(a'-a)$. This is impossible, unless $a'=a$. Such a solution must satisfy  $3^{i-a-b}2^w4^v5^u \geq l$, 
that is
\begin{eqnarray}\label{eqn_to_verify}
3^i \geq \frac{3^{a+b}l}{2^w4^v5^u}.
\end{eqnarray}
According to Theorem \ref{min_running_time_3_cases}, we have $(3/2)^al \leq 3^i < \min(3l,(3/2)^{a+1}l)$. Consequently, if we have 
\[\frac{3^{a+b}l}{2^w4^v5^u} \geq \min\left(3l,\left(\frac{3}{2}\right)^{a+1}l\right),\]
\[\textrm{where}\ \min\left(3l,\left(\frac{3}{2}\right)^{a+1}l\right)=\left\{
\begin{array}{ll}
  3l & \textrm{if}\ a=2\\
 (3/2)^{a+1}l & \textrm{if}\ a=0,1
\end{array}
\right.,\]
this solution does not meet the constraint (\ref{eqn_to_verify}), and then cannot exist. Among the 28 cases, we observe that 13 of them are not valid solutions.
Thus, we have 15 solutions, denoted $(u,v,w,a,b)$, for which we compute and sort the values $3^{a+b}l/(2^w4^v5^u)$ so that we can establish
their domains of validity. We then obtain the fifteen following cases:
\[
(\textrm{I})\left\{
\begin{array}{llrll}
|A|=i,    h_5=0,  h_4=0,  h_3=i,    h_2=0  & \hbox{ if } & l \leq & 3^i & < \frac{3l}{2}, \\
|A|=i,    h_5=0,  h_4=1,  h_3=i-2,  h_2=1 & \hbox{ if } & \frac{9l}{8} \leq & 3^i & < \frac{3l}{2}, \\
|A|=i-1,  h_5=0,  h_4=3,  h_3=i-4,  h_2=0 & \hbox{ if } & \frac{81l}{64} \leq & 3^i & < \frac{3l}{2}, \\
|A|=i-1,  h_5=1,  h_4=1,  h_3=i-3,  h_2=0 & \hbox{ if } & \frac{27l}{20} \leq & 3^i & < \frac{3l}{2}, \\
|A|=i-1,  h_5=0,  h_4=4,  h_3=i-6,  h_2=1 & \hbox{ if } & \frac{729l}{512} \leq & 3^i & < \frac{3l}{2}, \\
\end{array}
\right .
\]
\[
(\textrm{II})\left\{
\begin{array}{llrll}
|A|=i,    h_5=0,  h_4=0,  h_3=i-1,  h_2=1 & \hbox{ if } & \frac{3l}{2} \leq & 3^i & < \frac{9l}{4}, \\
|A|=i-1,  h_5=0,  h_4=2,  h_3=i-3,  h_2=0 & \hbox{ if } & \frac{27l}{16} \leq & 3^i & < \frac{9l}{4}, \\
|A|=i-1,  h_5=1,  h_4=0,  h_3=i-2,  h_2=0 & \hbox{ if } & \frac{9l}{5} \leq & 3^i & < \frac{9l}{4}, \\
|A|=i-1,  h_5=0,  h_4=3,  h_3=i-5,  h_2=1 & \hbox{ if } & \frac{243l}{128} \leq & 3^i & < \frac{9l}{4}, \\
|A|=i-1,  h_5=1,  h_4=1,  h_3=i-4,  h_2=1 & \hbox{ if } & \frac{81l}{40} \leq & 3^i & < \frac{9l}{4}, \\
|A|=i-2,  h_5=0,  h_4=5,  h_3=i-7,  h_2=0 & \hbox{ if } & \frac{2187l}{1024} \leq & 3^i & < \frac{9l}{4}, \\
\end{array}
\right .
\]
\[
(\textrm{III})\left\{
\begin{array}{llrll}
|A|=i-1,  h_5=0,  h_4=1,  h_3=i-2,  h_2=0 & \hbox{ if } & \frac{9l}{4} \leq & 3^i & < 3l, \\
|A|=i-1,  h_5=0,  h_4=2,  h_3=i-4,  h_2=1 & \hbox{ if } & \frac{81l}{32} \leq & 3^i & < 3l, \\
|A|=i-1,  h_5=1,  h_4=0,  h_3=i-3,  h_2=1 & \hbox{ if } & \frac{27l}{10} \leq & 3^i & < 3l, \\
|A|=i-2,  h_5=0,  h_4=4,  h_3=i-6,  h_2=0 & \hbox{ if } & \frac{729l}{256} \leq & 3^i & < 3l, \\
\end{array}
\right .
\]
In accordance with Theorem \ref{min_running_time_3_cases}, we have this grouping: 
\begin{itemize}
 \item Group I consists of the five cases ensuring a running time of $3i$;
 \item Group II consists of the six cases ensuring a running time of $3i-1$;
 \item Group III consists of the four cases ensuring a running time of $3i-2$.
\end{itemize}
Now, we have to optimize the arities. In the first group, we delete the last case since it decreases $h_5$ compared to the immediately preceding case. 
For the same reasons, we delete in the second group the fourth and sixth cases. Again, in the last group, we need to delete the last case. Overall, 11 cases are deduced
by intersecting the intervals of validity.
\end{proof}

\begin{remark}
The number of cases is lower when $l < 32$. Their determinations are let to the reader.
\end{remark}

\begin{remark}\label{local_optimality_rem}
Let us consider a tree which is optimal in the sense of the theorem \ref{min_running_time_max_arities}. If we extract a final subtree by deleting 
one or several lower levels (at the bottom), the resulting tree is still optimal. Indeed, let us suppose
that the original tree has height $h$ and has $l$ leaves (for $l$ message blocks). If we delete the $j$ lower levels, the resulting tree 
has $l'=\lceil l/(x_1x_2\cdots x_j) \rceil$ leaves and is already optimal for a number $l'$ of blocks.
If this form of local optimality does not exist, we can further optimize the original tree. Indeed, let us suppose that a $(x'_{j+1},x'_{j+2}, \ldots, x'_{h'})$-aries 
tree covers the number of blocks $l'$ and improves either the running time or the number of processors (in the sense of Theorem~\ref{min_running_time_max_arities}), 
compared to the $(x_{j+1},x_{j+2}, \ldots, x_h)$-aries tree. This means that a $(x_1, \ldots, x_j,x'_{j+1},x'_{j+2}, \ldots, x'_{h'})$-aries tree is a better choice 
to process $l$ leaves.
\end{remark}

Let us now consider trees having a number of leaves equal to the message length. 
Having a multiset of arities arranged in descending order, that we denote \break $A=\{x_1,x_2, \ldots, x_{|A|}\}$,
the number of nodes of level $i$ is $\lceil l/(x_1x_2 \ldots x_i) \rceil$. One important thing is the number of nodes 
of the base level. We have the following Corollary:

\begin{corollary}\label{numb_proc}
Let the message size be $l \geq 2$ and let $i$ be the lowest integer such that $3^i \geq l$. 
The number of processors required to process
such a message is:
\begin{itemize}
 \item $\lceil l/3 \rceil$ if $3^i \in [ l, \frac{9l}{8} [ \cup [ \frac{3l}{2}, \frac{27l}{16} [$,
 \item $\lceil l/4 \rceil$ if $3^i \in [ \frac{9l}{8} , \frac{27l}{20} [
 \cup [ \frac{27l}{16} , \frac{9l}{5} [ \cup [ \frac{9l}{4} , \frac{27l}{10} [$,

 \item $\lceil l/5 \rceil$ if $3^i \in [ \frac{27l}{20} , \frac{3l}{2} [
 \cup [ \frac{9l}{5} , \frac{9l}{4} [ \cup [ \frac{27l}{10} , 3l [$.
\end{itemize}

\end{corollary}

\begin{proof}
These results follow immediately from Theorem \ref{min_running_time_max_arities}.
\end{proof}
~\\

%\begin{theorem}\label{proba_15_cas}
%Let a message size $l$ drawn uniformly at random from the set $\llbracket 2, L \rrbracket$ where $L$ is a fixed positive integer. Let $k=\lfloor \log_3 L \rfloor$, 
%$i=\lceil \log_3 l \rceil$, and $\alpha$, 
%$\beta$ two real numbers such that $1 \leq \alpha < \beta < 3$. The probability that $3^i$ is in the interval $[ \alpha l, \beta l [$ is equal to 
%$$\frac{1}{L-1}\left(\frac{3}{2}\left( \frac{1}{\alpha} - \frac{1}{\beta} \right)(3^k - 1) + \eta \right)$$
%where $\eta=\left\{
%\begin{array}{ll}
%  3^{k+1}/\alpha - 3^{k+1}/\beta & \textrm{if}\ L \geq 3^{k+1}/\alpha\\
%  L - 3^{k+1}/\beta & \textrm{if}\ 3^{k+1}/\alpha > L \geq 3^{k+1}/\beta\\
%  0 & \textrm{otherwise}
%\end{array}
%\right..$
%\end{theorem}

An other important thing is the minimization of the \textit{total work} done by the hash tree algorithm for the processing of a single message. 
%This amount of \textit{work} corresponds to the total number of children of the nodes in the tree.
Since we are interested in hash trees having an optimal running time for a given message size $l$, 
we apply Theorem \ref{min_running_time_3_cases} or \ref{min_running_time_max_arities} to retrieve a topology.
For a perfect $(x_1,x_2, \ldots, x_h)$-aries tree constructed thanks to this theorem, the \textit{total work} is:
\[W_l=x_h+x_hx_{h-1}+x_hx_{h-1}x_{h-2} + \cdots + x_hx_{h-1}\ldots x_{2}x_1.\]
We notice that $\lceil \lceil \cdots \lceil \lceil l/x_1 \rceil / x_2 \rceil \cdots \rceil/ x_i \rceil = \lceil l/(x_1x_2 \cdots x_i) \rceil$ for (strictly) 
positive integers $(x_j)_{j=1 \ldots i}$. Consequently, for a $(x_1,x_2, \ldots, x_h)$-aries truncated tree constructed thanks to this theorem, the \textit{total work} is:
\[W_{tr,l}=l + \lceil l/x_1 \rceil + \lceil l/(x_1x_2) \rceil + \cdots + \lceil l/(x_1x_2\ldots x_{h-1}) \rceil.\]

This quantity is necessarily greater than or equal to $l$. Regarding truncated trees minimizing the running time, 
Theorem \ref{min_running_time_max_arities} indicates a topology which minimizes the \emph{total work}, by first choosing $x_1$ as large as possible, 
then choosing $x_2$ as large as possible, and so on.

\begin{remark}
 Decreasing the total work consists in decreasing as much as possible the number of nodes (apart from the number of leaves). We have to
 check if the multisets provided by Theorem \ref{min_running_time_max_arities} are preferable to others.
 Let us suppose that, for a given message size $l$, we have a (non-ordered) multiset minimizing
 the running time.
 Let us also suppose  that at least one order of this multiset minimizes the total work.
 %and whose at least one order minimizes the amount of work. 
 %%Clearly, if we arrange in descending order the arities of this multiset, we
 %%obtain a solution given by Theorem~\ref{min_running_time_max_arities}. 
 We can show that, among all the possible orders, this is the one represented in decreasing order
 which minimizes the total work. 
 We denote such a solution by $(x_1,x_2, \ldots, x_h)$ with $x_i \geq x_{i+1}$.
 Indeed, for a given message size $l$ and for any random permutation $\pi$ of the indices $\llbracket 1,h \rrbracket$, we have 
 $x_1 x_2 \ldots x_i  \geq x_{\pi(1)} x_{\pi(2)} \ldots x_{\pi(i)}$ for all $i \leq h$. Thus, summing left sides and right sides from the $h$
 inequalities, we have $\sum_{j=1}^i \lceil l/(x_1 x_2 \ldots x_j) \rceil \leq \sum_{j=1}^i \lceil l/(x_{\pi(1)} x_{\pi(2)} \ldots x_{\pi(j)}) \rceil$.
 When this ordered multiset cannot be derived from Theorem~\ref{min_running_time_max_arities}, we show that the transformations performed in its proof
 can further decrease the total work. We recall that the composition of five types of transformation 
 can lead to the eleven cases of this theorem. These transformations change the following pairs of arities $(3,3)$, $(2,2)$, $(3,2)$, $(4,3)$ and $(4,4)$ 
 into $(4,2)$, $(4)$, $(5)$, $(5,2)$ and $(5,3)$ respectively.
 It is sufficient to show that each of these transformations can reduce the total work:
 \begin{itemize}
  \item Case $(2,2) \rightarrow (4)$: We indeed have $\lceil l/4 \rceil \leq \lceil l/2 \rceil + \lceil l/(2 \cdot 2) \rceil$.
  \item Case $(3,2) \rightarrow (5)$: We indeed have $\lceil l/5 \rceil \leq \lceil l/3 \rceil + \lceil l/(3 \cdot 2) \rceil$.
  \item Case $(3,3) \rightarrow (4,2)$: Since $\lceil l/4 \rceil + \lceil l/8 \rceil \leq l/4 + l/8 + 2$ 
  and $l/3 + l/9 \leq \lceil l/3 \rceil + \lceil l/(3 \cdot 3) \rceil$, we can show that
  $\lceil l/4 \rceil + \lceil l/(4 \cdot 2) \rceil \leq \lceil l/3 \rceil + \lceil l/(3 \cdot 3) \rceil$
  for $l \geq 29$. As regards the other values of $l$, it appears that this inequality does not hold
  for $l=9$, but we remark that this message length is not concerned by such a transformation.
  \item Case $(4,3) \rightarrow (5,2)$: By the same reasoning, we can show that 
  $\lceil l/5 \rceil + \lceil l/(5 \cdot 2) \rceil \leq \lceil l/4 \rceil + \lceil l/(4 \cdot 3) \rceil$ for $l \neq 11$ and $l \neq 12$.
  We then remark that a message of length $l=11$ or $12$ blocks cannot be covered by a $(5,2)$-aries tree.
  \item Case $(4,4) \rightarrow (5,3)$: Again, we can show that 
  $\lceil l/5 \rceil + \lceil l/(5 \cdot 3) \rceil \leq \lceil l/4 \rceil + \lceil l/(4 \cdot 4) \rceil$ for $l \neq 16$.
  We remark that a length of $l=16$ cannot be covered by a $(5,3)$-aries tree.
 \end{itemize}

\end{remark}

\section{About the distribution of cases}\label{prob_dist}

For the purpose of minimizing the number of processors at each step of the computation, we apply Theorem~\ref{min_running_time_max_arities}. 
There are 11 possible cases and we would like to estimate their distribution. We then perform an empirical estimation with a Monte Carlo simulation.
Numerous papers have analysed the distribution of file sizes, and more particularly the sizes of files transferred accross a network. 
It is shown that the average size of transferred files is about 10KB \cite{TMW97,WAWB05,BC98}, and we observe that
the Pareto distribution
is the predominantly used model \cite{BC98,Dow05,WAWB05,HSBA2003} to fit the datasets, with a \emph{shape} parameter $\rho$ generally estimated between $1$ and $2$,
and a \emph{location} parameter $\nu$ estimated by hand or as a function of the average file size (using the mean formula). For this simulation, 
%by assuming a $32$-byte block size, 
we use the Pareto model with $\rho=1.5$ and 
%a \emph{location} parameter 
$\nu=10^4 \cdot (\rho - 1)/\rho$ bytes in order to generate a sample of $10^6$ byte sizes.
We initialize 11 counters $cpt_i=0$ for $i=1, \ldots, 11$. 
With the assumption of a $32$-byte block size, we perform the following operations for each generated $\textrm{byte-size}$:
\begin{enumerate}
 \item We compute $v=3^{\lceil \log(\lceil \textrm{byte-size}/32 \rceil )/\log 3 \rceil}$;
 \item We check which one of the following cases is satisfied:
 Case~$1$:~$l \leq v < \frac{9l}{8}$;
 Case~$2$:~$\frac{9l}{8} \leq v < \frac{81l}{64}$;
 Case~$3$:~$\frac{81l}{64} \leq v < \frac{27l}{20}$;
 Case~$4$:~$\frac{27l}{20} \leq v < \frac{3l}{2}$;
 Case~$5$:~$\frac{3l}{2} \leq v < \frac{27l}{16}$;
 Case~$6$:~$\frac{27l}{16} \leq v < \frac{9l}{5}$;
 Case~$7$:~$\frac{9l}{5} \leq v < \frac{81l}{40}$;
 Case~$8$:~$\frac{81l}{40} \leq v < \frac{9l}{4}$;
 Case~$9$:~$\frac{9l}{4} \leq v < \frac{81l}{32}$;
 Case~$10$:~$\frac{81l}{32} \leq v < \frac{27l}{10}$;
 Case~$11$:~$\frac{27l}{10} \leq v < 3l$;~\\
 The index of the satisfied case is denoted $j$;
 \item We increment $cpt_j$.
\end{enumerate}
The relative frequency of the $i$-th case is then $cpt_i/10^6$ for $i=1, \ldots, 11$.
These estimations are made using R software \cite{RCTEAM} with the VGAM library \cite{Yee10}.
According to Corollary~\ref{numb_proc},
it follows that the number of required processors is $\lceil l/3 \rceil$ with probability 
of about $19.1\%$, $\lceil l/4 \rceil$
with probability of about $33.1\%$, and $\lceil l/5 \rceil$ with 
probability of about $47.8\%$.
The relative proportions of each individual case are depicted in Figure \ref{barplot_11_cases}.
\begin{figure}[h]
 \centering
 \includegraphics[scale=0.45]{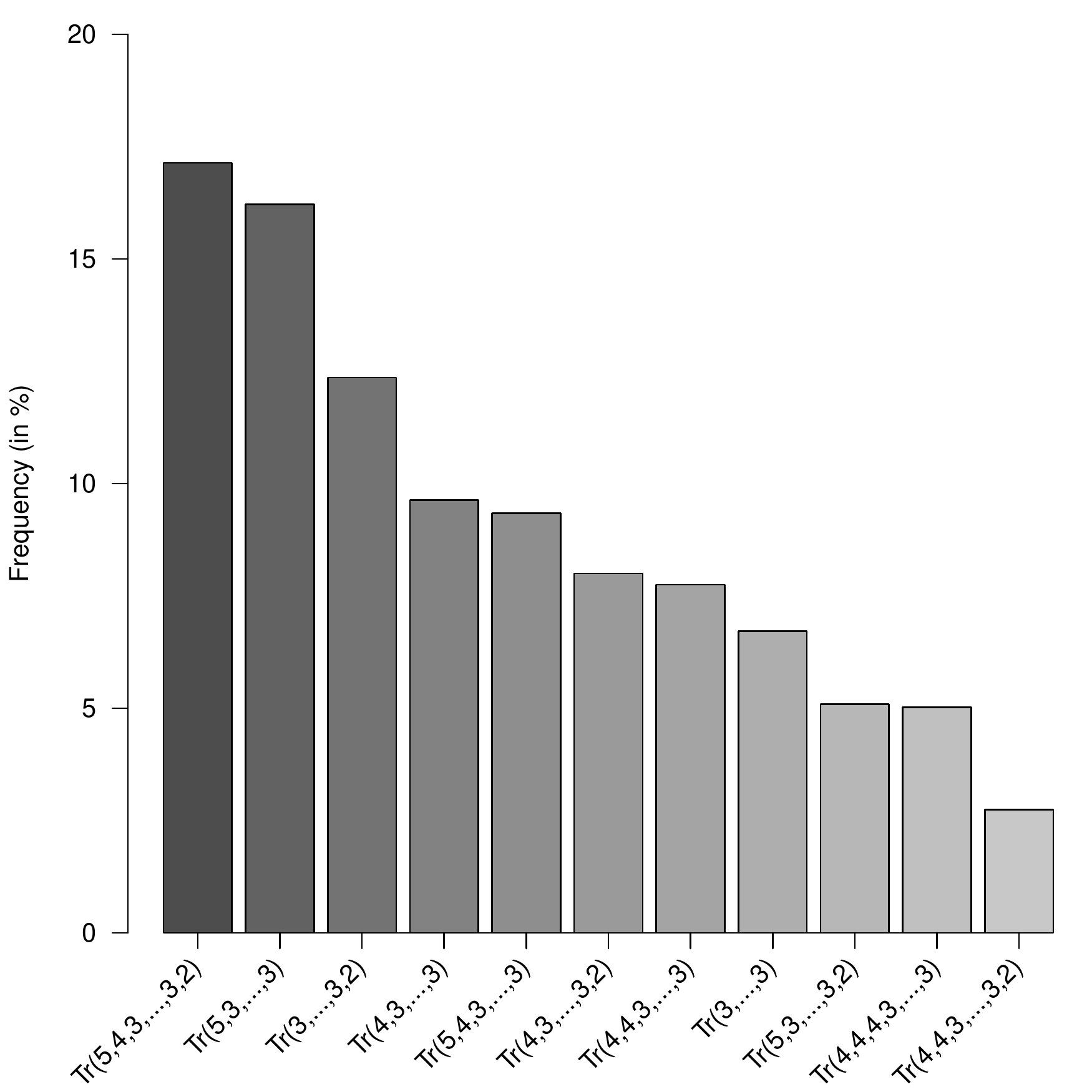}
 \caption{Proportions for the eleven cases of Theorem \ref{min_running_time_max_arities}. The bars are drawn in decreasing order of frequency.
 The notation Tr($a_1$,$a_2$,...,$a_h$) stands for a tree having arities $a_1$, $a_2$, ..., $a_h$ from the base level to the root node.}
 \label{barplot_11_cases}
\end{figure}

\section{Applying our optimizations safely}\label{sec_section2}

Suppose that we have 4 different inner functions $f_{BL}$, $f_{I}$, $f_{F}$ and $f_{SN}$ with the following properties:
\begin{itemize}
 \item $f_{t} : \{0,1\}^{qN} \rightarrow \{0,1\}^N$ for $q \geq 2$ and $t \in \{BL,$ $I,$ $F,$ $SN\}$.
 \item Without precomputation, they have the same running time of $q+2$ units of
 %they are rate-1 and 
  time when compressing $q$ message blocks of size $N$ bits.
  \item With a constant number of precomputations, they have the same running time of $q$ units of time when
  compressing $q$ message blocks.
 %(\emph{i.e.} they have a running time of exactly $y$ units of time when compressing $y$ blocks of size $N$ bits),
 \item They behave like independent random oracles.
\end{itemize}
In the hash tree construction we propose, we use $f_{BL}$ to compute base level nodes, $f_{I}$ to compute inner nodes, $f_{F}$
to compute the root node. If the tree is of height one, there is only one node computed using $f_{SN}$. In order
to simulate four independent functions, we use the same inner function $f$ but with domain separation. Indeed, since $f$ behaves like a random oracle,
by construction the functions $f_{BL}(x) = f(100^{N-3} \| x)$, $f_{I}(x) = f(000^{N-3} \| x)$, $f_{F}(x) = f(010^{N-3} \| x)$ and  $f_{SN}(x) = f(110^{N-3} \| x)$
behave like independent random oracles.
~\\

Since $f$ is based on the \emph{prefix-free} MD construction, the first block encodes the number of blocks of the message, comprised between $2$ and $5$. 
Due to the domain separation codings,
the second block contains the bitstring $b_1b_20^{N-3}m_{11}$ where the 1-bit values $b_1$ and $b_2$ depend on the type of processed input, 
and $m_{11}$ is the first bit of the input. Overall, we have to precompute 16 hash states (resulting from the processing of these two blocks) in order to
obtain an inner function having a running time equals to the node arity.

\paragraph{On the precomputations costs induced by the use of a sponge function}
The first block merely consists of the bitstring $b_1b_20^{N-4}m_{11}m_{12}$ where the 1-bit values $b_1$ and $b_2$ are as above,
and $m_{11}$ and $m_{22}$ are the first two bits of the message input. Again, we have to precompute 32 hash states (resulting from the processing of only one block) 
in order to obtain an inner function having a running time equal to the node arity.

\subsection{An example of hash function}
%behave like independent random oracles, 
Given a message $M$, a hash tree mode could be the following:
\begin{enumerate}
 \item Whatever the message bit-length is, we append to $M$ a bit "1" and the minimum number of bits "0" so that the total bit-length is a multiple
 of $N$. The new message is denoted $M_0$ and its total number of blocks of size $N$ bits is $l=|M_0|_2/N$, where
 $|M_0|_2$ is the bitlength of $M_0$. 
 %We split the padded message, denoted $M_0$, into blocks $M_{0,1}$, $M_{0,2}$, ..., $M_{0,l}$ of size $N$ bits, with $l=\lfloor |M_0|_2/N \rfloor+1$.
 \item We apply Theorem \ref{min_running_time_max_arities} on $l$ to retrieve the height 
 $h$ of the tree and an ordered multiset $A$ of arities $a_1$, $a_2$, ..., $a_h$ (arranged in decreasing order).
 \emph{This step requires the computation of $\log_3 l$, which can be considered free.}
 \item If $h=1$, we compute and return the hash value $f_{SN}(M_0)$. Otherwise, we go to the following step.
 \item We first split $M_0$ into blocks $M_{0,1}$, $M_{0,2}$, ..., $M_{0,l_1}$ where: ($i$) $l_1=\lceil l/a_1 \rceil$; ($ii$)~all blocks but the last one are
$a_1N$ bits long and the last block is between $N$ and $a_1N$ bits long.
 Then, we compute the message
 \[M_1:= \bigparallel_{j=1}^{l_1} f_{BL}\left(M_{0,j}\right).\]
 \item If $h=2$ we go to step 6. Otherwise, for $k=2$ to $h-1$, we perform the following operations:
 \begin{enumerate}
  \item We split $M_{k-1}$ into blocks $M_{k-1,1}$, $M_{k-1,2}$, ..., $M_{k-1,l_k}$ where: ($i$)~$l_k=\left\lceil l/\left(\prod_{j=1}^k a_j\right) \right\rceil$; 
  ($ii$)~all blocks but the last one are $a_kN$ bits long and the last block is between $N$ and $a_kN$ bits long.
  \item We compute the message 
  \[M_k:= \bigparallel_{j=1}^{l_k} f_{I}\left(M_{k-1,j}\right).\]
 \end{enumerate}
  \item We compute and return the hash value $f_{F}(M_{h-1})$.
 %Then, compute a last block $M''_{\lceil l/a_1 \rceil}=f_{BaseLevel}(M'_{a_1(\lceil l/a_1 \rceil-1)+1}\| \cdots \| M'_{l})$.

\end{enumerate}

\subsection{Security}\label{sec_section}

% Since these four functions behave like independent random oracles, then by construction a function
% $f$ such that $f_{BL}(x) = f(10 \| x)$, $f_{I}(x) = f(00 \| x)$,
% $f_{F}(x) = f(01 \| x)$ and  $f_{SN}(x) = f(11 \| x)$ also behaves 
% like a random oracle. So if we make that assumption on our four inner functions,
% we can go ahead and analyze our mode rewritten as if it was using only $f$
% and check whether the three conditions are satisfied.~\\

Since our mode can be rewritten as if it was using only $f$,
it suffices to check whether the three conditions (seen in Section \ref{three_conditions}) are satisfied to prove soundness.~\\

\noindent
We first need to describe some rules regarding our mode:

\paragraph{Rule 0} The root $f$-input has a prepended code $010^{N-3}$ or $110^{N-3}$. 

\paragraph{Rule 1} A $f$-input with a prepended code $010^{N-3}$ has children having a prepended code $0^{N-1}$ or $10^{N-2}$. 

\paragraph{Rule 2} A $f$-input with a prepended code $10^{N-2}$ or $110^{N-3}$ has no children.

\paragraph{Rule 3} A $f$-input with a prepended code 00 has children having a prepended code $10^{N-2}$ or $0^{N-1}$.

\paragraph{Rule A} A $f$-input must be $(2+kN)$-bit long with an integer $k \geq 1$.

\paragraph{Rule B} At a same level of the tree, the number of chaining values is the same for all the $f$-inputs, except for the righmost one where 
 this number may be smaller.
 
%\paragraph{Rule C} The leaves are at the same depth.
\paragraph{Rule C} At a same level of the tree, prepended codes are the same for all the $f$-inputs.~\\

\noindent
Note that the satisfaction of the rules 0, 1, 2, 3 and C imply that the leaves are at the same depth. So, we do not need to define a rule to express this.~\\

By construction our mode
is final-$f$-input separable. Our mode is trivially message-complete since it processes all message bits. Indeed, 
%we are able, 
%given 
having the valuated tree of $f$-inputs produced by $\tau$, the algorithm $A_{message}$ reaches directly the base level $f$-inputs and recovers 
message blocks by discarding the frame bits, whether they serve as padding purpose (in the rightmost $f$-input) or for identifying the type of $f$-input.  
This algorithm runs in linear time in the number of bits in the tree.
Finally, our mode is also tree-decodable. Thanks to domain separation between base level $f$-inputs and other $f$-inputs, 
we cannot find a tree 
%of $f$-inputs 
which is both \emph{compliant} and \emph{final-subtree-compliant}.
%we cannot find a complete tree of $f$-inputs generated 
%by the tree hash mode $\tau$ whose characteristics (its topology, the size of each $f$-input, and the frame bits) are identical to those of a final proper subtree 
%(of an other tree) generated by $\tau$. 
Given only one $f$-input, the prepended coding allows its content to be recognized correctly. We can then construct a decoding algorithm $A_{decode}$ which runs 
in $2$ phases: 
Phase $1$ starts from the root $f$-input and fully determines the tree structure by recursively decoding each $f$-input. The size of a $f$-input determines the 
number of its children. This phase terminates with the ``correct'' state C0 if all the visited $f$-inputs respect the rules defined above.
%and if the last $f$-inputs that are reached have all a prepended code starting with the bit $1$. 
This phase terminates with the ``incorrect'' state C1 if one of the rule 0, 2, A, B or C is not respected. 
If Rule 1 is not respected it terminates with the ``incorrect'' state C2. Otherwise, it terminates with the ``incorrect'' state C3.
Phase $2$ examines the properties of 
the decoded tree by taking into account the termination state of the first phase. The details of Phase 2 are the followings:
\begin{enumerate}
%  \item At any point, if a prepended code in a $f$-inputs is not consistent with the mode 
% (for instance, a tree having more than one $f$-input, where the root $f$-input starts with the bit $1$), or if a $f$-input is not $(2+kN)$-bit 
% long with an integer $k \geq 2$, then it returns ``incompliant''. 
%  \item Once the tree has been entirely and correctly decoded (\emph{i.e.} the last encountered $f$-inputs, at the base level, have all a prepended code starting with the bit $1$), 
%  some examinations are made: 
\item If the state is C0, it runs the $A_{message}$ algorithm in order to check the message size. If for the corresponding number $l$ of blocks, 
 Theorem \ref{min_running_time_max_arities} indicates a topology which differs from
 the one that Phase 1 has just decoded, then it returns ``incompliant''. Otherwise, it returns ``compliant''.

%  \item If the prepended codes starting with the bit $1$ have all been reached and the prepended codes in the $f$-inputs are all consistent, then the tree is presumed
%  to be compliant. It runs the $A_{message}$ algorithm 
%  in order to check the message size. If for the corresponding number $l$ of blocks, Theorem 3 indicates a topology which differs from 
%  the one it has just decoded, then, again, it returns ``incompliant''.
%  \item If 
 \item If the state is C1, it returns ``incompliant''.
 \item If the state is C2, the following examinations are made:
 \begin{enumerate}
  \item If the tree seems incomplete with a single $f$-input, it checks, after having discarded the prepended code, the number of blocks of size $N$ bits.
  If there are 2, 3 or 4 blocks, then it returns ``final-subtree-compliant''. Otherwise, it returns ``incompliant''.
  \item Otherwise, the coding is incompatible with the mode, and it returns ``incompliant''.
 \end{enumerate}
 \item If the state is C3, this means that only the rule 3 %and/or C 
 is not respected. The following examinations are then made:
 \begin{enumerate}
  \item If there is a $f$-input with a prepended code 00 which has a child with a prepended code not equal to 10, then it returns ``incompliant''.
%   \item Otherwise, there is at least one missing node. For each level $i$, it counts the maxmum number $p_i$ of blocks that 
%   can be contained in a $f$-input. It computes the product $l=\prod_{i=1}^{h} p_i$ and it checks what topology
%   Theorem 3 indicates for this number. If the arity of each level indicated by this theorem is consistant with the numbers $p_i$, then it returns 
%   ``final-subtree-compliant''. Otherwise, it returns ``incompliant''.
  \item Otherwise, there is at least one missing $f$-input. We note $h$ the maximum number of $f$-inputs visited on a path in this proper final subtree. 
  %Thus, $h$ corresponds to its height using the first convention for the definition of a node.
  The algorithm has to check if its topology is consistent with Theorem\footnote{Meaning that it has to detect if this final subtree can be extracted 
  from an optimal tree, in the sense of this theorem.}~\ref{min_running_time_max_arities}. 
  First, it establishes a system of contraints regarding the arity of each level.  If there is only one $f$-input at a level which is not the root, 
  and if it is the righmost one\footnote{A $f$-input is the righmost one in the complete tree if we see that it does not have a right sibling, when looking at the 
  retrieved topology by Phase 1.} in the complete tree, then the number of (message or chaining) blocks it contains defines a lower bound for the arity 
  of this level. Otherwise, the arity for this level has a single possible value and the constraint is an equality. Having these $h$ constraints, it checks in 
  Theorem \ref{min_running_time_max_arities}
  which cases (among the eleven) can satisfy these contraints. We denote by $L$ the list of compatible cases and for each case $j$ in $L$ we denote by $a_i$ the arity 
  of the level $i$. For each case $j$ in $L$, the algorithm performs the following operations:
  \begin{enumerate}
   \item It completes this final subtree until its leaves are at the same depth $h$. It performs this 
   by (virtually) creating the missing $f$-inputs with a maximal arity (\emph{i.e.}, even if a missing $f$-input at level $i$ is the rightmost one 
   in the complete tree, it chooses its arity to be exactly $a_i$).
   \item It counts the number $l$ of blocks covered by the completed subtree. If $l$ is in the domain of validity of the case $j$, then it returns ``final-subtree-compliant''.
  \end{enumerate}
  A this point, no cases in $L$ are suitable. This phase 2 finally returns ``incompliant''.
 \end{enumerate}

\end{enumerate}

\noindent
The total running time of $A_{decode}$ is linear in the number of bits in the tree.

% \paragraph{Remark 1.}According to Bertoni \textit{et al.} \cite{BDPV14_Suf} (Section 7.5), if the mode
%  ensures \emph{message-completeness} and
%  if we are able to prevent making inner nodes out to be base level nodes (by means of domain separation)
%  %, \textit{e.g.}, using special-purpose frame bits) 
%  then the mode ensures \emph{tree-decodability} as well.

\begin{remark}If prepending 2 bits is sufficient \cite{BDPV14_Suf} for soundness, we could have used a reduction to the Sakura coding \cite{BDPV14_Sak} 
where meta-information bitstrings are longer. 
Using Sakura coding allows any tree-based hash function
to be automatically indifferentiable from a random oracle, without the need of further proofs.
According to the Sakura ABNF grammar \cite{BDPV14_Sak}, the number of chaining values (\textit{i.e.} the number of children of a $f$-input) 
is also coded in the formatted input to $f$.  
In our context and using this coding, the number of ways to format the input to $f$ (depending on 
its location in the tree topology) is at least~10. A Sakura coding bit 
expresses the fact that 
the input is or is not the last one (for the computation of the root hash). Another bit indicates whether the input contains a message block
or a certain number of chaining values (whose number, 2, 3, 4 or 5, is also coded at the beginning of the input). 
Overall, this corresponds to 10 ways to format an input in our 
tree topology. This larger encoding and the fact to have meta-information bits at the beginning (for prepended bits) 
and at the end of the input (for appended bits) both complicates our construction and increases the number of precomputed hash states.
%Thus, we would have used 16 inner functions instead of 4 (since we have 4 possible number of childrens and 4 kind of nodes.
\end{remark}

\begin{remark}
 Yet another solution is to use different IVs (Initial Values) instead of particular frame bits, as suggested in \cite{BDPV14_Suf,luk13}. We could use a \emph{free-IV} hash function,
 like the \emph{suffix-free-prefix-free} hash function from Bagheri \emph{et al.} \cite{BGKZ12}. The distinction between the $f$-inputs would be done by using 4 distinct 
 IVs: \texttt{BL\_IV} for base level $f$-inputs, \texttt{I\_IV} for inner $f$-inputs, \texttt{F\_IV} for the root $f$-input and \texttt{SN\_IV} for a tree reduced to
 a single $f$-input.
\end{remark}

\section{Conclusion}

In this paper, we focused on trees having their leaves at the same depth. 
We have shown, for a given message length, 
how to construct a hash tree minimizing the running time. 
In particular, we have shown how to minimize the
number of processors allowing such a running time. 
The proposed construction makes use of a prepended 
coding for each input to the inner function in order to satisfy 
the three conditions of Bertoni \textit{et al.} \cite{BDPV14_Suf}.
Besides, our tree topologies could also be used in substitution of
the tree hash mode of Skein, provided that the tweaks to 
the UBI mode are carefully chosen for each node.

% We have also seen that it is possible to slightly decrease 
% the number of processors by considering other types of trees. 
% Analysis on few small message sizes have revealed that, in the best case, we can save one more processor by using a tree which does not have all its leaves
% at the same depth. 
% Further work is necessary to adequately specify to what extent the amount of resources
% can actually be decreased.

\bibliographystyle{plain}
\bibliography{trees}

\begin{thebibliography}{10}

\bibitem{AMV88a}
Gordon~B. Agnew, Ronald~C. Mullin, and Scott~A. Vanstone.
\newblock Fast exponentiation in \emph{GF(2\({}^{\mbox{n}}\))}.
\newblock In {\em Advances in Cryptology - {EUROCRYPT} '88, Workshop on the
  Theory and Application of of Cryptographic Techniques, Davos, Switzerland,
  May 25-27, 1988, Proceedings}, pages 251--255, 1988.

\bibitem{AMP10}
Elena Andreeva, Bart Mennink, and Bart Preneel.
\newblock Security reductions of the second round {SHA-3} candidates.
\newblock Cryptology ePrint Archive, Report 2010/381, 2010.

\bibitem{AMP10isc}
Elena Andreeva, Bart Mennink, and Bart Preneel.
\newblock Security reductions of the second round {SHA-3} candidates.
\newblock In {\em Information Security - 13th International Conference, {ISC}
  2010, Boca Raton, FL, USA, October 25-28, 2010, Revised Selected Papers},
  pages 39--53, 2010.

\bibitem{ANWW13}
Jean-Philippe Aumasson, Samuel Neves, Zooko Wilcox-O'Hearn, and Christian
  Winnerlein.
\newblock {BLAKE2}: Simpler, smaller, fast as {MD5}.
\newblock In {\em Proceedings of the 11th International Conference on Applied
  Cryptography and Network Security}, ACNS'13, pages 119--135, Berlin,
  Heidelberg, 2013. Springer-Verlag.

\bibitem{BGKZ12}
Nasour Bagheri, Praveen Gauravaram, Lars~R. Knudsen, and Erik Zenner.
\newblock The suffix-free-prefix-free hash function construction and its
  indifferentiability security analysis.
\newblock {\em International Journal of Information Security}, 11(6):419--434,
  2012.

\bibitem{BC98}
Paul Barford and Mark Crovella.
\newblock Generating representative web workloads for network and server
  performance evaluation.
\newblock In {\em Proceedings of the 1998 ACM SIGMETRICS Joint International
  Conference on Measurement and Modeling of Computer Systems}, SIGMETRICS
  '98/PERFORMANCE '98, pages 151--160, New York, NY, USA, 1998. ACM.

\bibitem{BerDaePeeVan13}
Guido Bertoni, Joan Daemen, Micha{\"{e}}l Peeters, and Gilles~Van Assche.
\newblock Keccak.
\newblock In {\em Advances in Cryptology - {EUROCRYPT} 2013, 32nd Annual
  International Conference on the Theory and Applications of Cryptographic
  Techniques, Athens, Greece, May 26-30, 2013. Proceedings}, pages 313--314,
  2013.

\bibitem{BerDaePeeVan08}
Guido Bertoni, Joan Daemen, Micha\"{e}l Peeters, and Gilles Van~Assche.
\newblock On the indifferentiability of the sponge construction.
\newblock In {\em Proceedings of the Theory and Applications of Cryptographic
  Techniques 27th Annual International Conference on Advances in Cryptology},
  EUROCRYPT'08, pages 181--197, Berlin, Heidelberg, 2008. Springer-Verlag.

\bibitem{BDPV09}
Guido Bertoni, Joan Daemen, Michael Peeters, and Gilles {Van Assche}.
\newblock Sufficient conditions for sound tree and sequential hashing modes.
\newblock Cryptology ePrint Archive, Report 2009/210, 2009.

\bibitem{BDPV14_Suf}
Guido Bertoni, Joan Daemen, Micha\"{e}l Peeters, and Gilles Van~Assche.
\newblock Sufficient conditions for sound tree and sequential hashing modes.
\newblock {\em Int. J. Inf. Secur.}, 13(4):335--353, August 2014.

\bibitem{BDPV14_Sak}
Guido Bertoni, Joan Daemen, Michaël Peeters, and Gilles Van~Assche.
\newblock Sakura: A flexible coding for tree hashing.
\newblock 8479:217--234, 2014.

\bibitem{CLNY06}
Donghoon Chang, Sangjin Lee, Mridul Nandi, and Moti Yung.
\newblock Indifferentiable security analysis of popular hash functions with
  prefix-free padding.
\newblock In {\em Proceedings of the 12th International Conference on Theory
  and Application of Cryptology and Information Security}, ASIACRYPT'06, pages
  283--298, Berlin, Heidelberg, 2006. Springer-Verlag.

\bibitem{CDMP05}
Jean{-}S{\'{e}}bastien Coron, Yevgeniy Dodis, C{\'{e}}cile Malinaud, and
  Prashant Puniya.
\newblock Merkle-damg{\aa}rd revisited: How to construct a hash function.
\newblock In {\em Advances in Cryptology - {CRYPTO} 2005: 25th Annual
  International Cryptology Conference, Santa Barbara, California, USA, August
  14-18, 2005, Proceedings}, pages 430--448, 2005.

\bibitem{Dam90}
Ivan Damg{\aa}rd.
\newblock A design principle for hash functions.
\newblock In {\em CRYPTO '89: Proceedings of the 9th Annual International
  Cryptology Conference on Advances in Cryptology}, pages 416--427, London, UK,
  1990. Springer-Verlag.

\bibitem{Dow05}
Allen~B. Downey.
\newblock {L}ognormal and {P}areto distributions in the {I}nternet.
\newblock {\em Comput. Commun.}, 28(7):790--801, May 2005.

\bibitem{FLSWBKCW09}
Niels Ferguson, Stefan~Lucks Bauhaus, Bruce Schneier, Doug Whiting, Mihir
  Bellare, Tadayoshi Kohno, Jon Callas, and Jesse Walker.
\newblock The skein hash function family (version 1.2), 2009.

\bibitem{GiRy88}
Alan Gibbons and Wojciech Rytter.
\newblock {\em Efficient parallel algorithms}.
\newblock Cambridge University Press, 1988.

\bibitem{Gue14}
Shay Gueron.
\newblock Parallelized hashing via j-lanes and j-pointers tree modes, with
  applications to {SHA-256}.
\newblock {\em {IACR} Cryptology ePrint Archive}, 2014:170, 2014.

\bibitem{GK12a}
Shay Gueron and Vlad Krasnov.
\newblock Parallelizing message schedules to accelerate the computations of
  hash functions.
\newblock {\em J. Cryptographic Engineering}, 2(4):241--253, 2012.

\bibitem{GK12b}
Shay Gueron and Vlad Krasnov.
\newblock Simultaneous hashing of multiple messages.
\newblock {\em J. Information Security}, 3(4):319--325, 2012.

\bibitem{HSBA2003}
Mor Harchol-Balter, Bianca Schroeder, Nikhil Bansal, and Mukesh Agrawal.
\newblock Size-based scheduling to improve web performance.
\newblock {\em ACM Trans. Comput. Syst.}, 21(2):207--233, May 2003.

\bibitem{LKPC05}
Mun{-}Kyu Lee, Yoonjeong Kim, Kunsoo Park, and Yookun Cho.
\newblock Efficient parallel exponentiation in gf(qn) using normal basis
  representations.
\newblock {\em J. Algorithms}, 54(2):205--221, 2005.

\bibitem{luk13}
Stefan Lucks.
\newblock Tree hashing: {A} simple generic tree hashing mode designed for
  {SHA-2} and {SHA-3}, applicable to other hash functions.
\newblock In {\em Early Symmetric Crypto (ESC)}, 2013.

\bibitem{Mer80}
Ralph~C. Merkle.
\newblock Protocols for public key cryptosystems.
\newblock In {\em Proceedings of the 1980 {IEEE} Symposium on Security and
  Privacy}, pages 122--134, 1980.

\bibitem{Mer79}
Ralph~Charles Merkle.
\newblock {\em Secrecy, Authentication, and Public Key Systems.}
\newblock PhD thesis, Stanford, CA, USA, 1979.

\bibitem{PS03}
Pinakpani Pal and Palash Sarkar.
\newblock {PARSHA-256-} - {A} new parallelizable hash function and a
  multithreaded implementation.
\newblock In {\em Fast Software Encryption, 10th International Workshop, {FSE}
  2003, Lund, Sweden, February 24-26, 2003, Revised Papers}, pages 347--361,
  2003.

\bibitem{PGV93}
Bart Preneel, Ren{\'e} Govaerts, and Joos Vandewalle.
\newblock Hash functions based on block ciphers: A synthetic approach.
\newblock In {\em Proceedings of the 13th Annual International Cryptology
  Conference on Advances in Cryptology}, CRYPTO '93, pages 368--378, London,
  UK, UK, 1994. Springer-Verlag.

\bibitem{RCTEAM}
{R Core Team}.
\newblock {\em R: A Language and Environment for Statistical Computing}.
\newblock R Foundation for Statistical Computing, Vienna, Austria, 2014.

\bibitem{RABCDEKKLRSSSTY08}
Ronald~L. Rivest, Benjamin Agre, Daniel~V. Bailey, Christopher Crutchfield,
  Yevgeniy Dodis, Kermin Elliott, Fleming~Asif Khan, Jayant Krishnamurthy,
  Yuncheng Lin, Leo Reyzin, Emily Shen, Jim Sukha, Drew Sutherland, Eran
  Tromer, and Yiqun~Lisa Yin.
\newblock The {MD6} hash function: A proposal to nist for sha-3, 2008.

\bibitem{SS01}
Palash Sarkar and Paul~J. Schellenberg.
\newblock A parallel algorithm for extending cryptographic hash functions.
\newblock In {\em Progress in Cryptology - {INDOCRYPT} 2001, Second
  International Conference on Cryptology in India, Chennai, India, December
  16-20, 2001, Proceedings}, pages 40--49, 2001.

\bibitem{SS02}
Palash Sarkar and Paul~J. Schellenberg.
\newblock A parallelizable design principle for cryptographic hash functions.
\newblock {\em {IACR} Cryptology ePrint Archive}, 2002:31, 2002.

\bibitem{Sti90}
Douglas~R. Stinson.
\newblock Some observations on parallel algorithms for fast exponentiation in
  {GF($2^n$)}.
\newblock {\em {SIAM} J. Comput.}, 19(4):711--717, 1990.

\bibitem{TMW97}
K.~Thompson, G.~J. Miller, and R.~Wilder.
\newblock Wide-area internet traffic patterns and characteristics.
\newblock {\em Netwrk. Mag. of Global Internetwkg.}, 11(6):10--23, November
  1997.

\bibitem{Gat91}
Joachim von~zur Gathen.
\newblock Efficient exponentiation in finite fields (extended abstract).
\newblock In {\em 32nd Annual Symposium on Foundations of Computer Science, San
  Juan, Puerto Rico, 1-4 October 1991}, pages 384--391, 1991.

\bibitem{WAWB05}
Adepele Williams, Martin Arlitt, Carey Williamson, and Ken Barker.
\newblock {\em Web Workload Characterization: Ten Years Later}, pages 3--21.
\newblock Springer US, Boston, MA, 2005.

\bibitem{WLLC06}
Chia{-}Long Wu, Der{-}Chyuan Lou, Jui{-}Chang Lai, and Te{-}Jen Chang.
\newblock Fast parallel exponentiation algorithm for {RSA} public-key
  cryptosystem.
\newblock {\em Informatica, Lith. Acad. Sci.}, 17(3):445--462, 2006.

\bibitem{Yee10}
Thomas~W. Yee.
\newblock The {VGAM} package for categorical data analysis.
\newblock {\em Journal of Statistical Software}, 32(10):1--34, 2010.

\end{thebibliography}

\newpage

\appendix

\section{Comparison between a perfect binary tree and a perfect ternary tree}\label{comp}
Let $l\geq 2$ an integer.
Let $h_2$ the lowest integer such that $2^{h_2} \geq l$ and $h_3$ the lowest integer such that $3^{h_3} \geq l$.
We assume that we use a perfect binary (or ternary) tree as in the original Merkle (and Damg{\aa}rd) hash tree mode,
\emph{i.e.} the message is padded to obtain a message size which is a power of $2$ (or $3$).
The problem is to compare $2h_2$ and $3h_3$.

Any $l$ can be uniquely written
\[l=2^k+u,\]
where $u$ is an integer such that $0 \leq u <2^k$. Then
\[l= 2^k(1+a) \hbox{ where } a=\frac{u}{2^k}.\]
If $a=0$ then $h_2=k$ else $h_2=k+1$.

\subsection{The case $a=0$}
In this case 
\[ l=2^k, h_2=k,  
h_3=\left\lceil \frac{k\log(2)}{\log(3)}\right\rceil.\]
Then 
\[h_3=\frac{k\log(2)}{\log(3)}+\alpha,\]
where $0<\alpha <1$.
We must compare 
$3h_3$ with $2 h_2$, namely
\[3\frac{k\log(2)}{\log(3)}+3 \alpha \hbox{ with } 2k,\]
or
\[3 \frac{\log(2)}{\log(3)}+3 \frac{\alpha}{k} \hbox{ with } 2.\]
As $\alpha$ is bounded by $1$ and $3 \dfrac{\log(2)}{\log(3)} <2$, for $k$ sufficiently large
we have $3h_3<2h_2$.
More precisely if $k \geq 28$ then $3h_3<2h_2$, meaning that a perfect ternary tree gives a better running time than a perfect binary tree.
When $2\leq k \leq 27$, we compute the $27$ values 
\[T=\frac{3}{k}\left\lceil \frac{k\log(2)}{\log(3)}\right\rceil - 2\]
and we look at the sign of the result: 
\begin{itemize}
\item For $k=3 s$ ($s=1, \cdots,9$), a perfect binary tree and a perfect ternary tree give the same result ($T=0$).
\item For $k=11,14,17,19,20,22,23,25,26$,  a perfect ternary tree is better ($T<0$).
\item For $k=2,4,5,7,8,10,13,16$, a perfect binary tree is better ($T>0$).
\end{itemize}

\subsection{The case $a \neq 0$}
In this case $h_2=k+1$ and 
\[h_3= \left\lceil \frac{k\log(2)}{\log(3)}+\frac{\log(1+a)}{\log(3)}\right\rceil.\]
We must compare $3h_3$ to $2h_2$.
But:
\[\frac{3h_3}{k} \leq \frac{3\log(2)}{\log(3)}+\frac{3\log(2)}{k\log(3)} +\frac{3}{k}\]
and 
\[\frac{2h_2}{k}=2+\frac{2}{k}.\]
As $\dfrac{3\log(2)}{\log(3)}<2$, for $k$ sufficiently large
we have $3h_3<2h_2$.
More precisely for $k \geq 27$ then $3h_3 <2h_2$, meaning that a perfect ternary tree gives a better running time than a 
perfect binary tree. For any $2 \leq k \leq 26$ and any $u$
such that $1 \leq u < 2^k$ we must compute the sign of
\[R=3\left\lceil \frac{k\log(2)}{\log(3)}+\frac{\log\left(1+\frac{u}{2^k}\right)}{\log(3)}\right\rceil-2k-2.\]
As $R$ is an increasing function of $u$, it is sufficient to determine for any $k<27$
the value of $u$ where the sign changes. This can be done by dichotomy. Results are in Table~\ref{Compar_perfect_binary_ternary}.

\begin{table}
\[
\begin{array}{|c|c|}
\hline 
k=2 & Sign=0 \hbox{ for any } u \\
\hline
k=3 & Sign <0 \hbox{ for } u=1 \hbox{ and } Sign > 0 \hbox{ for } u >1\\
\hline
k=4 & Sign <0 \hbox{ for } u \leq 11 \hbox{ and } Sign > 0 \hbox{ for } u >11\\
\hline
k=5 & Sign=0 \hbox{ for any } u \\
\hline 
k=6 & Sign <0 \hbox{ for } u \leq 17 \hbox{ and } Sign > 0 \hbox{ for } u >17\\
\hline
k=7 & Sign <0 \hbox{ for } u \leq 115 \hbox{ and } Sign > 0 \hbox{ for } u >115\\
\hline
k=8 & Sign=0 \hbox{ for any } u \\ 
\hline
k=9 & Sign <0 \hbox{ for } u \leq 217 \hbox{ and } Sign > 0 \hbox{ for } u >217\\
\hline
k=10 & Sign<0 \hbox{ for any } u \\ 
\hline
k=11 & Sign <0 \hbox{ for } u \leq 139 \hbox{ and } Sign = 0 \hbox{ for } u >139\\
\hline
k=12 & Sign <0 \hbox{ for } u \leq 2465 \hbox{ and } Sign > 0 \hbox{ for } u >2465\\
\hline
k=13 & Sign<0 \hbox{ for any } u \\ 
\hline
k=14 & Sign <0 \hbox{ for } u \leq 3299 \hbox{ and } Sign = 0 \hbox{ for } u >3299\\
\hline
k=15 & Sign <0 \hbox{ for } u \leq 26281 \hbox{ and } Sign > 0 \hbox{ for } u >26281\\
\hline
k=16 & Sign<0 \hbox{ for any } u\\
\hline
k=17 & Sign <0 \hbox{ for } u \leq 46075 \hbox{ and } Sign = 0 \hbox{ for } u >46075\\
\hline
k=18 & Sign <0 \hbox{ for any } u\\
\hline
k=19 & Sign <0 \hbox{ for any } u\\
\hline
k=20 & Sign <0 \hbox{ for } u \leq 545747 \hbox{ and } Sign = 0 \hbox{ for } u >545747\\
\hline
k=21 & Sign <0 \hbox{ for any } u\\
\hline
k=22 & Sign <0 \hbox{ for any } u\\
\hline
k=23 & Sign <0 \hbox{ for } u \leq 5960299 \hbox{ and } Sign = 0 \hbox{ for } u >5960299\\
\hline
k=24 & Sign <0 \hbox{ for any } u\\
\hline
k=25 & Sign <0 \hbox{ for any } u\\
\hline
k=26 & Sign <0 \hbox{ for } u \leq 62031299 \hbox{ and } Sign = 0 \hbox{ for } u >62031299\\
\hline
\end{array}
\]
\caption{Comparison between a perfect binary tree and a perfect ternary tree. If $Sign<0$ a perfect ternary tree has a better running time. 
If $Sign=0$ the two trees give the same running time. Otherwise a perfect binary tree is better.}
\label{Compar_perfect_binary_ternary}
\end{table}

\section{Algorithms for reducing the number of processors}\label{red_number_processors}

\subsection{Reducing the number of processors at the base level}

We propose two (different) algorithms to construct an optimal tree (in the sense of the running time) which covers exactly $l$ blocks (the tree is not necessarily perfect)
and increases as much as possible the arity of the base level.
The first solution consists to check if there exists an optimal tree having a level of arity $5$ or $4$.~\\

\noindent\fbox{%
\begin{varwidth}{\dimexpr\linewidth-2\fboxsep-2\fboxrule\relax}
\begin{sloppypar}
\paragraph{\textbf{Algorithm 2a}} ~\\
INPUTS: a message length $l$ and a multiset of arities
(arranged in descending order) minimizing the running time, denoted  $A=\{x_1,x_2, ..., x_{|A|}\}$.~\\
OUTPUT: a multiset of arities (still sorted in descending order)
minimizing the number of processors while leaving unchanged the running time.~\\
Let $t_l$ the optimal running time for a message of size $l$, \emph{i.e.} the sum of arities of $A$.
The algorithm proceeds as follows: 
\begin{enumerate}
 \item Use Algorithm 1 to construct a tree for a message length $l'=\lceil l/5 \rceil$ and denote by $A'$ the corresponding ordered multiset of arities.
 If $t_l = t_{l'} + 5$ then return the multiset $A''=\{5,A'\}$, otherwise go to the following step.
 \item Use Algorithm 1 to construct a tree for a message length $l'=\lceil l/4 \rceil$ and denote by $A'$ the corresponding ordered multiset of arities.
 If $t_l = t_{l'} + 4$ then return the multiset $A''=\{4,A'\}$, otherwise go to the following step.
 \item Return $A$ (which cannot be further optimized).
\end{enumerate}
\end{sloppypar}
\end{varwidth}
}

\newpage

The second approach uses the following hints:~\\

\paragraph{\textbf{Hints}} Let us note that if $k>0$, then $a>b \iff (a-k)b>a(b-k)$. Moreover, if $b \leq a$ then $(b-1)(a+1) \le ab$. This suggests that
a product of several numbers, where the sum is constant, is maximized when these numbers are as close together as possible.
In order to decrease the product of arities as slowly as possible we 
%are going to 
use the fact that if $c \geq b \geq a$ 
%then 
we have
$(c+1)(b-1)a \geq (c+1)b(a-1)$.~\\

\noindent\fbox{%
\begin{varwidth}{\dimexpr\linewidth-2\fboxsep-2\fboxrule\relax}
\begin{sloppypar}
\paragraph{\textbf{Algorithm 2b}} ~\\
INPUTS: a message length $l$ and a multiset of arities
(arranged in descending order) minimizing the running time, denoted  $A=\{x_1,x_2, ..., x_{|A|}\}$.~\\
OUTPUT: a multiset of arities (still sorted in descending order)
minimizing the number of processors while leaving unchanged the running time.~\\
% L’algorithme prend en entrée une taille de message $l$, un multi-ensemble d’arités ordonnées dans l’ordre 
% décroissant minimisant le temps d'exécution, noté $A = \{x_1 , x_2 , \ldots, x_{|A|} \}$, et retourne un multi-ensemble d’arités ordonnées 
% qui conserve ce temps d'exécution minimum et qui minimise le nombre de processeurs.
The algorithm proceeds as follows: 
\begin{enumerate}
 \item We start by replacing in $A$ each pair of arities $2$ by an arity $4$ (leaving possibly only one arity $2$ in $A$). 
 We sort $A$ in descending order.
 \item We repeat at most twice the following routine to determine the solution:
 \begin{itemize}
  \item Case $|A|=1$: we return $A$.
  \item Case $|A|=2$:
    \begin{itemize}
      \item Case $x_1=5$: we return $A$.
      \item Case $x_1 \geq 3,x_2 \geq 3$: if $(x_1+1)(x_2-1) \geq l$ then $A=\{x_1+1,x_2-1\}$, otherwise we return $A$.
      \item Case $x_1=4,x_2=2$: we return $A$.
      \item Case $x_1=3,x_2=2$: if $5 \geq l$ then $A=\{5\}$. We return $A$.
    \end{itemize}
  \item Case $|A| \geq 3$:
    \begin{itemize}
      \item Case $x_1=5$: we return $A$.
      \item Case $x_1 \geq 3,x_2 \geq 3,x_3 \geq 2$: if $(x_1+1)(x_2-1)\prod_{i=3}^{|A|}x_i \geq l$ then we perform the following operations: 
      ($i$)~we add $1$ to $x_1$ and we subtract $1$ to $x_2$; ($ii$)~we replace a possible pair of arities $2$ by an arity $4$; 
      ($iii$)~we reorder~$A$. If either the check fails or $x_1=5$ then we return $A$.
    \end{itemize}
 \end{itemize}
\end{enumerate}
\end{sloppypar}
\end{varwidth}
}

\newpage

\subsection{Reducing the number of processors at all the levels}

The following algorithm uses Algorithm 1 and 2 in order to compute a multiset of arities (sorted in descending order) 
minimizing the running time and the required number of processors at each step of the computation.~\\

\noindent\fbox{%
\begin{varwidth}{\dimexpr\linewidth-2\fboxsep-2\fboxrule\relax}
\paragraph{\textbf{Algorithm 3}} ~\\
INPUT: a message length $l$.~\\
OUTPUT: an ordered multiset of arities
minimizing the running time and the required number of processors at each step of the computation.~\\
Let $A_0=\{x_1,x_2, ..., x_{|A_0|}\}$ be the multiset of arities returned by Algorithm 1. 
We then use Algorithm 2 with a message of length $l$ and 
the multiset $A_0$ to compute the multiset of arities $A_1=\{x'_1,x'_2, ..., x'_{|A_1|}\}$. 
The rest of the algorithm proceeds iteratively as follows:
%L'algorithme se déroule itérativement de la façon suivante :
\begin{sloppypar}
\begin{enumerate}
 \item 
% On applique l'algorithme 2 avec les entrées $l$ et $A_0=\{x_1,x_2, ..., x_{|A_0|}\}$ et il nous retourne le multi-ensemble
% $A_1=\{x'_1,x'_2, ..., x'_{|A_1|}\}$. 
We apply Algorithm 2 on inputs $l'=\lceil l/x'_1 \rceil$ and $A'_1=\{x'_2, ..., x'_{|A_1|}\}$ to 
compute the multiset $A'_2=\{x''_2, ..., x''_{|A'_1|}\}$. We set $n=1$.
 \item As long as one of the following termination conditions is not met, namely
 $(i)$~$A^{(n)}_{n+1}=A^{(n)}_n$; ($ii$)~the highest number of levels of arity $4$ 
has been reached (see Lemma \ref{numb_ar_4_and_5}); or $(iii)$~$A^{(n)}_{n+1}=\varnothing$, 
we set $n=n+1$ and apply Algorithm 2 with the inputs $l^{(n)}=\left\lceil l^{(n-1)}/x^{(n)}_n \right\rceil$ and 
$A^{(n)}_n=\{x^{(n)}_{n+1}, \ldots., x^{(n)}_{|A^{(n-1)}_n|}\}$ to compute the multiset 
$A^{(n)}_{n+1}=\{x^{(n+1)}_{n+1}, \ldots, x^{(n+1)}_{|A^{(n)}_n|}\}$.
 
%Tant que l'une des trois conditions suivantes n'est pas satisfaites, à savoir $(i)$ $A^{(n)}_{n+1}=A^{(n)}_n$; ($ii$) Le plus grand nombre possible d'étages d'arité $4$ 
%a été atteint (voir lemme $5$); ou bien $(iii)$ $A^{(n)}_{n+1}=\varnothing$, on applique l'algorithme 2 avec les entrées $l^{(n)}=\left\lceil l^{(n-1)}/x^{(n)}_n \right\rceil$ et 
%$A^{(n)}_n=\{x^{(n)}_{n+1}, \ldots., x^{(n)}_{|A^{(n-1)}_n|}\}$ et il nous retourne le multi-ensemble $A^{(n)}_{n+1}=\{x^{(n+1)}_{n+1}, \ldots, x^{(n+1)}_{|A^{(n)}_n|}\}$.
%  \item On applique l'algorithme 2 avec les entrées $l^{(n)}=l^{(n-1)}/x^{(n)}_n$ et $A^{(n)}_n=\{x^{(n)}_{n+1}, ..., x^{(n)}_{|A^{(n-1)}_n|}\}$ et 
% il nous retourne le multi-ensemble $A^{(n)}_{n+1}=\{x^{(n+1)}_{n+1}, ..., x^{(n+1)}_{|A^{(n)}_n|}\}$. 
% La boucle s'arrête de tourner dès que l'une des
% conditions suivantes est satisfaite~:
%   \begin{itemize}
%     \item $A^{(n)}_{n+1}=A^{(n)}_n$;
%     \item Le plus grand nombre possible d'étages d'arité $4$ a été atteint (voir lemme $5$);
%     \item $A^{(n)}_{n+1}=\varnothing$.
%   \end{itemize}
\end{enumerate}
The resulting multiset of arities $A_r=\{x'_1, x''_2, \ldots, x^{(n)}_n, x^{(n+1)}_{n+1}, \ldots, x^{(n+1)}_{|A^{(n)}_n|} \}$
minimizes the number of required processors at each step of the computation.
\end{sloppypar}
% Il en résulte alors un multi-ensemble d'arités minimisant le nombre de processeurs, 
% que l'on note $A_r=\{x'_1, x''_2, \ldots, x^{(n)}_n, x^{(n+1)}_{n+1}, \ldots, x^{(n+1)}_{|A^{(n)}_n|} \}$.
\end{varwidth}
}

\end{document}